\documentclass[journal]{IEEEtran}
\usepackage[english]{babel}
\usepackage{amsmath,amsfonts}
\usepackage{amsthm}

\newtheorem{lemma}{Lemma}
\newtheorem{remark}{Remark}
\usepackage{array}
\usepackage{textcomp}
\usepackage{stfloats}
\usepackage{url}
\usepackage{verbatim}
\usepackage{graphicx}
\def\BibTeX{{\rm B\kern-.05em{\sc i\kern-.025em b}\kern-.08em
    T\kern-.1667em\lower.7ex\hbox{E}\kern-.125emX}}
\usepackage{balance}
\usepackage{amssymb}
\usepackage{optidef}

\usepackage{xcolor}
\usepackage{hyperref}
\usepackage{cleveref}
\usepackage{lipsum}
\usepackage{epsfig} 
\usepackage{balance}
\usepackage{textgreek}
\usepackage{algpseudocode}
\usepackage[linesnumbered,ruled,vlined]{algorithm2e}
\usepackage{nomencl}
\usepackage{subfigure}
\usepackage{cite}

\usepackage{soul,color}

\begin{document}
\title{Quantum Annealing-Based Sum Rate Maximization for Multi-UAV-Aided Wireless Networks}

\author{Seon-Geun Jeong, Pham Dang Anh Duc, Quang Vinh Do, Dae-Il Noh, Nguyen Xuan Tung, Trinh Van Chien, Quoc-Viet~Pham,~\IEEEmembership{Senior Member,~IEEE}, Mikio Hasegawa,~\IEEEmembership{Member,~IEEE}, Hiroo Sekiya,~\IEEEmembership{Senior Member,~IEEE}, and Won-Joo Hwang,~\IEEEmembership{Senior Member,~IEEE}

\thanks{Seon-Geun Jeong, Dae-Il Noh and Nguyen Xuan Tung are with the Department of Information Convergence Engineering, Pusan National University, Pusan National University, Busan 46241, South Korea (e-mail: wjdtjsrms11@pusan.ac.kr, nohdi1991@pusan.ac.kr), tung.nguyenxuan1310@pusan.ac.kr).}%

\thanks{Quang Vinh Do is with the Wireless Communications Research Group, Faculty of Electrical and Electronics Engineering, Ton Duc Thang University, Ho Chi Minh City, Vietnam (e-mail: dovinhquang@tdtu.edu.vn).}%

\thanks{Hiroo Sekiya is with the Graduate School of Engineering, Chiba University, Chiba, 263-8522, Japan (e-mail: sekiya@faculty.chiba-u.jp).}%

 \thanks{Trinh Van Chien and Pham Dang Anh Duc  are with the School of Information and Communication Technology (SoICT), Hanoi University of Science and Technology (HUST), 100000 Hanoi, Vietnam (email: chientv@soict.hust.edu.vn,
duc.pda210207@sis.hust.edu.vn)}

\thanks{Quoc-Viet~Pham is with the School of Computer Science and Statistics, Trinity College Dublin, Dublin, D02PN40, Ireland (e-mail: viet.pham@tcd.ie).}

\thanks{Mikio Hasegawa is with the Department of Electrical Engineering, Tokyo University of Science, Tokyo, 162-8601, Japan (e-mail: hasegawa@haselab.ee.kagu.tus.ac.jp).}%

\thanks{Won-Joo Hwang (corresponding author) is with the School of Computer Science and Engineering, Center for Artificial Intelligence Research, Pusan National University, Busan 46241, South Korea (e-mail: wjhwang@pusan.ac.kr).}%

}
\maketitle

\begin{abstract}
In wireless communication networks, it is difficult to solve many NP-hard problems owing to computational complexity and high cost. Recently, quantum annealing (QA) based on quantum physics was introduced as a key enabler for solving optimization problems quickly. However, only some studies consider quantum-based approaches in wireless communications. Therefore, we investigate the performance of a QA solution to an optimization problem in wireless networks. Specifically, we aim to maximize the sum rate by jointly optimizing clustering, sub-channel assignment, and power allocation in a multi-unmanned aerial vehicle-aided wireless network. We formulate the sum rate maximization problem as a combinatorial optimization problem. Then, we divide it into two sub-problems: 1) a QA-based clustering and 2) sub-channel assignment and power allocation for a given clustering configuration. Subsequently, we obtain an optimized solution for the joint optimization problem by solving these two sub-problems. For the first sub-problem, we convert the problem into a simplified quadratic unconstrained binary optimization (QUBO) model. As for the second sub-problem, we introduce a novel QA algorithm with optimal scaling parameters to address it. Simulation results demonstrate the effectiveness of the proposed algorithm in terms of the sum rate and running time.
\end{abstract}


\section{Introduction}
\label{sec:introduction}
\IEEEPARstart{U}{nmanned} aerial vehicles (UAVs) have been deployed in a wide range of communication scenarios to enhance network coverage and performance owing to their deployment flexibility, versatility, and cost-effectiveness \cite{oubbati2020softwarization,ullah2020cognition}.
Benefiting from line-of-sight (LoS) propagation, UAV-aided communications have emerged for wireless connectivity \cite{7470933}. UAVs can serve as aerial base stations (ABSs) to support high spectral-efficient communications in underserved areas \cite{8660516}. UAVs can also help recover essential communication services in disaster areas where the terrestrial base station (BS) is damaged \cite{8466046,8669870}.
They can operate as relays for connecting dead zone users with the BSs and expanding the coverage area \cite{8068199}. 
Moreover, different UAV-based aerial platforms for wireless services have attracted tremendous research attention from academic and industrial communities \cite{9584850,zeng2016throughput,zhang2019cellular,yang2018three}. Despite these promising advantages, deploying UAVs for wireless communications poses many practical challenges \cite{nguyen2018real, 7936601, 8789457, 8038869,9003500,9662406,8760424, shang2020spectrum, wang2018enabling}.

Resource allocation and real-time control are the two main challenges in the design employment of UAV-aided networks \cite{8038869,9003500,9662406,8760424,7936601, 8789457, nguyen2018real}.
Conventionally, UAV communication is implemented on unlicensed spectrum, such as industry and science bands. With the tremendous growth in mobile devices and high-throughput applications, the unlicensed spectrum cannot guarantee communication performance. This issue comes up with the idea of sharing licensed spectrum bands to support high-quality connectivity for UAVs and increasing on-demand services for users \cite{wang2018enabling}. However, sharing spectrum resources also gives rise to extra interference, especially in practical multi-UAV-aided networks where handling co-channel interference is much more complicated than in terrestrial systems \cite{8760424}. Specifically, in a network provided by UAVs, as more
UAVs are deployed to provide services to ground users (GUs),
communication between these UAVs are often deployed on unlicensed frequency bands. These unlicensed frequency bands do
not guarantee communication performance due to frequency
resource limitations. Consequently, UAVs have to operate
on the same frequency, causing their signals to overlap, which will cause signal interference leading to co-channel interference. Moreover, the service area of each UAV impacts the channel gain between the ground users (GUs) and the UAVs. When the number of GUs increases, the UAVs' limited spectrum resources and transmit power lead to low throughput and high traffic load \cite{7936601, 8789457}. Therefore, an efficient method is required to mitigate co-channel interference.
Efficient real-time control and operation is an important aspect of UAV communications due to its lifetime and dynamic environment \cite{nguyen2018real}. Thus, many optimization problems need to be solved in a short time. 
Numerous studies focused on performance optimization for UAV-aided networks, e.g., energy saving \cite{yang2018energy}, time and power optimization \cite{wang2017resource}, and secure communication \cite{zhou2019secure}. 
 Therefore, obtaining optimal solutions is difficult for most of these problems owing to computation complexity and high cost. 

Quantum computing has exhibited potential applications in wireless communications by employing hardware-based search algorithms, such as \textit{D-Wave} \cite{dwave} and coherent ising machine (CIM) \cite{takesue2020simulating}, to enable fast optimization of many large-scale problems based on Ising models. \textit{D-Wave} quantum annealing machines (QAMs) are designed to solve optimization problems using stochastic Hamiltonians. Starting from \textit{D-Wave One} with 128 qubits released in 2011, the \textit{D-Wave Advantage2} prototype with 7000 qubits is being developed. Recently, \textit{D-Wave} has introduced a hybrid solver service (HSS) that contains three solvers, including the binary quadratic model (BQM), discrete quadratic model (DQM), and constrained quadratic model (CQM). These solvers can be applied to combinatorial optimization problems defined on discrete or continuous variables \cite{dwave}.
The BQM and DQM solvers read unconstrained quadratic problems defined on a maximum of 1,000,000 binary and 5,000 discrete variables, respectively. The most general CQM solver can be used for problems defined on a maximum of 500,000 binaries, integers, and real numbers. It is worth noting that the CQM solver is the best method for representing and solving constrained optimization problems \cite{dwave}. Every solver in the HSS has a classical front-end that reads an input and a time limit. It invokes one or more hybrid heuristic solvers to search for good-quality solutions for the given input. The heuristic solvers run in parallel on state-of-the-art CPU or GPU platforms provided by Amazon web services (AWS). In such platforms, a classical heuristic module explores the solution space, while a quantum module formulates quantum queries that are sent to a back-end Advantage QPU. 
While \textit{D-Wave's} HSS does not solely rely on quantum computing, the integration of quantum solution methods can lead to a phenomenon called hybrid acceleration. This unique hybrid workflow can find better solutions at a faster rate compared to purely classical approaches.

Quantum annealing (QA), a heuristic method based on quantum physics for solving the complex combinatorial optimization problem, has appeared in various disciplines \cite{dwave}.  
Using QA methods, original optimization problems can be converted into quadratic unconstrained binary optimization (QUBO) or Ising models, with penalty and scaling parameters \cite{9065239}. The fundamental rule is to use quantum physics to find the lowest energy states of the problems, which are made by the optimal or approximately optimal combinations of spins. The main advantage of this method is that it can obtain the solutions quickly by using QAMs. However, a clear answer to the speedup question is still being investigated, which remains an active research area \cite{RevModPhys.90.015002}. Moreover, QA is a promising tool that provides fast solutions to many optimization problems in wireless communications. However, no existing research considers quantum-based approaches for clustering and resource allocation in multi-UAV-aided networks. 

The main contributions of this paper can be summarized as follows:
\begin{itemize}
    \item We develop a novel QA framework to enhance the sum rate for multi-UAV-aided wireless networks. To achieve this, we aim to maximize the sum rate by jointly optimizing user clustering, sub-channel assignment, and power allocation for the downlink under the influence of co-channel interference\footnote{{Multiple access method such as non-orthogonal multiple access (NOMA), is also applied to solve co-channel interference, but the complexity of this method makes it less suitable for wireless networks with limited resources. NOMA enhances spectral efficiency by allowing multiple users to share the same time and frequency resources when overlapping signals. However, this technique also significantly increases computational demands, making it impractical for low-power IoT devices \cite{noma}. We acknowledge the potential of this approach in other scenarios, but this method is out of our topic, and we plan to explore its applicability under different system models in future research.}}. The sum rate maximization problem is formulated as a combinatorial optimization problem, which is then transformed into a QUBO model. To the best of our knowledge, this is the first attempt at using a QA method for optimization problems in multi-UAV-aided wireless networks.
    \item To solve this optimization problem, we divide it into two sub-problems. The first is an optimization of QA-based user clustering, and the second is the joint optimization of sub-channel assignment and power allocation for given clusters. We transform the problems into QUBO models 
    and obtain the approximately optimal solution using a \textit{D-Wave}'s QAM.
    \item Solving these QUBO models is still challenging because the model's scaling and penalty parameters must be determined in advance. Therefore, we present a method to derive the lower and upper bounds for penalty parameters. Moreover, we employ the Taylor series approximation and Mixed-integer linear fractional programming (MILFP) method to derive optimal scaling parameters. As far as we know, our paper is the first to apply this method for finding the scaling parameters of the QUBO model. Until now, most of the previous papers only adjust the scaling parameters by using heuristic methods.
    \item Finally, we provide comprehensive simulation results to verify the efficiency of the proposed solution in terms of the clustering, sum rate, and computational time. Specifically, our method improves clustering with poor-matching to approximately \(0\% \), outperforming benchmark methods. In addition, the sum rate is also significantly improved when achieving better results than benchmark methods in the considered scenarios. Finally, the running time of our method is much faster and remains almost constant while the running time of compared methods increases with the size of the problem.
\end{itemize}

The rest of this paper is organized as follows. In Section \ref{sec_system}, we present the system model of a multi-UAV-aided wireless network and formulate the sum rate maximization problem as a combinatorial optimization problem. In Section \ref{sec_solution}, a joint optimization method is proposed to obtain the locally optimal solution by dividing the maximization problem into two sub-problems, where we transform the sub-problems into a QUBO model with Taylor series expansion and the MILFP method. Simulation results are presented in Section \ref{Simulationresult}, and finally, we conclude the paper in Section \ref{sec_conclusion}.

\section{Problem Description}
\label{sec_system}
\subsection{System Model}
\renewcommand{\figurename}{{Fig.}}
\begin{figure}[ht]
\centering
\includegraphics[width=\linewidth]{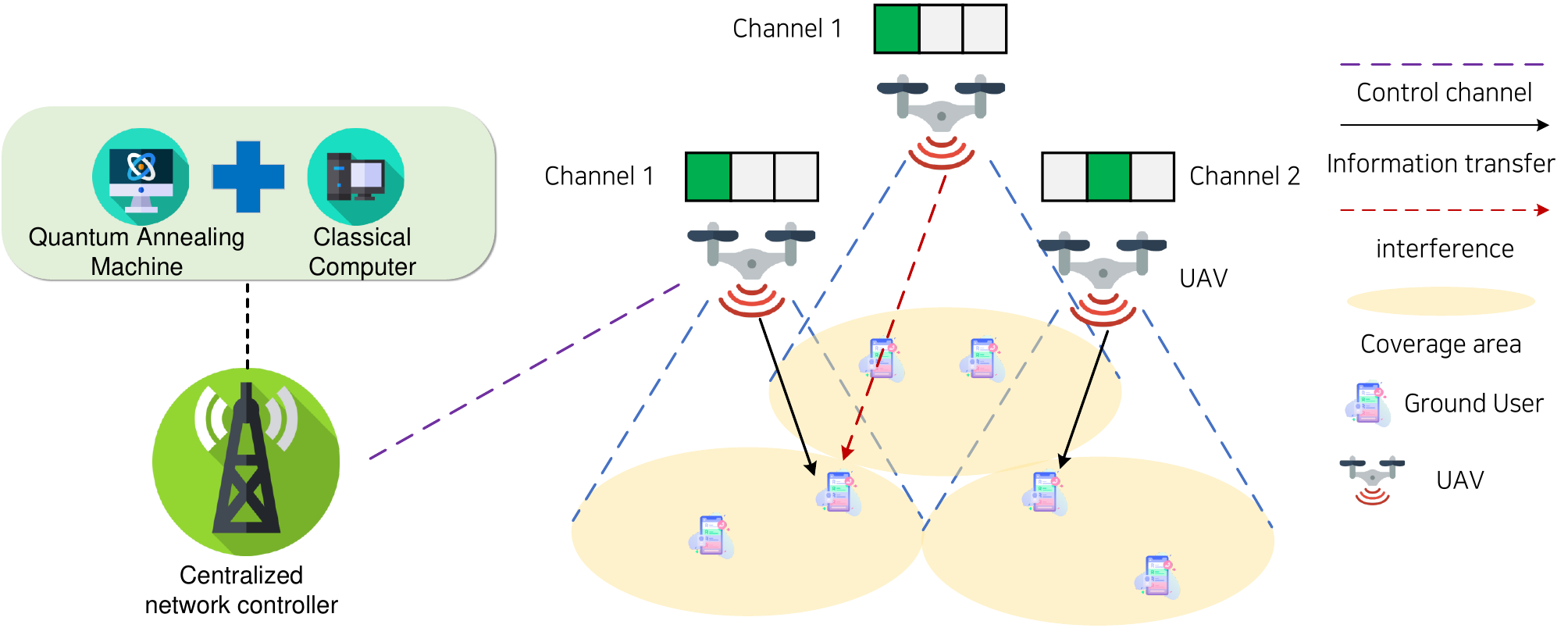}
\caption{System model.}
\label{fig:systemmodel}
\end{figure}
We consider a multi-UAV-aided wireless network where a set $\mathcal{M}$ of $M$ single-antenna UAVs are deployed as ABSs to provide data services to single-antenna GUs on an optimized trajectory{
\footnote{{
      We acknowledge the importance of UAV's trajectory for increasing the efficiency of data transmission to GUs \cite{9750200}, \cite{10185601}. Joint optimization of trajectory and other variables optimize the solution, however trajectory is the master variable so it has high impact on the solutions of other variables. Solving optimization problems at each coordinate of the UAV can also be expanded in future work when simultaneously considering the trajectory of the UAV as well as other variables in the model. From there, we consider scenarios assuming that the trajectory of the UAV has been optimized and predetermined throughout the operation, and we focus on sub-problems such as clustering optimization, subchannel assignment, and power control.}}
} over a specific duration \cite{9126869}. The duration time is divided into $T$ time slots, denoted as $\mathcal{T} = \{t_1,t_2,\dots,t_T\}$ and we consider the problem on each of these time slots. We assume that the system operation is managed by a centralized network controller located in the ground station (GS) that frequently sends control signals to UAVs and GUs via control channels. { Specifically, the controller takes advantage of the computing power of a quantum computer, using a QPU instead of a CPU to analyze bandwidth allocation and power control for downlink data transmission from UAVs to GUs.}
Each UAV $m \in \mathcal{M}$ hovers at a constant altitude $H$ to serve a set $\mathcal{N}_{m}$ of GUs that are randomly distributed within the UAV's coverage area. The set of all GUs in the system is denoted by $\mathcal{N} \triangleq \cup_{m=1}^{M} \mathcal{N}_{m}$. 
The total system bandwidth is divided into $K$ sub-channels, denoted as $\mathcal{K} = \{1,2,\dots,K\}$. It is noteworthy that sub-channels occupied by UAVs may overlap with each other.

The horizontal locations of UAV $m \in \mathcal{M}$ and GU $n \in \mathcal{N}$ are denoted by $\boldsymbol{u}_{m} = [x_{m},y_{m}]^T$ and $\boldsymbol{v}_{n} = [x_{n},y_{n}]^T$, respectively. Thus, the distance between them can be calculated as $d_{m,n} =\sqrt{\| \boldsymbol{u}_{m} -\boldsymbol{v}_{n} \| ^{2} + H^2}$. 
In the considered system, both line-of-sight (LoS) and non-line-of-sight (NLoS) links are possible. The GUs can face NLoS links with UAVs owing to uncertain obstacles, such as buildings and trees. Therefore, we consider a probabilistic path-loss model, in which the UAV-to-GU communication channels can be modeled as either LoS or NLoS links. The LoS and NLoS path-loss between UAV $m$ and GU $n$ is given as follows \cite{6863654}: 
\begin{equation}
g_{m,n}^{LoS} = 20\log(4\pi f_c d_{m,n}/c) + \eta^{LoS},
\end{equation}
\begin{equation}
g_{m,n}^{NLoS} = 20\log(4\pi f_c d_{m,n}/c) + \eta^{NLoS},
\end{equation}
where $f_c$ and $c$ are the carrier frequency and microwave speed, respectively. $\eta^{LoS}$ and $\eta^{NLoS}$ are additional signal-attenuation factors for LoS and NLoS, respectively.
In this context, the LoS probability $\rho^{LoS}_{m,n}$ for the link between UAV $m$ and GU $n$ is given as
\begin{equation}
\rho _{m,n}^{LoS} =\frac{1}{1+ae^{-b\left(\frac{180}{\pi }\arcsin\left(\frac{H}{d_{m,n}}\right) -a\right)}},
\end{equation}
where $a$ and $b$ are environment-dependent parameters. In addition, the NLoS probability is expressed by $\rho_{m,n}^{NLoS}=1-\rho_{m,n}^{LoS}$. Accordingly, the probabilistic path-loss between UAV $m$ and GU $n$ can be calculated as follows: 
\begin{equation}
g_{m,n} = g_{m,n}^{LoS} \rho_{m,n}^{LoS} +g_{m,n}^{NLoS} \rho_{m,n}^{NLoS}.
\end{equation}

The sub-channel assignment variable is denoted as $\beta^{k}_{m} \in \{0,1\}$, which indicates that channel $k$ is allocated to UAV $m$ if $\beta^{k}_{m}=1$; otherwise, $\beta^{k}_{m}=0$. We assume that each UAV can only occupy one sub-channel, as follows:
\begin{equation}
\sum _{k\in \mathcal{K}} \beta^{k}_{m} \leq 1,\ \forall m \in \mathcal{M}.
\end{equation}
We also denote the user-association variable as $s_{m,n} \in \{0,1\}$, which indicates that GU $n$ is associated with UAV $m$ if $s_{m,n} = 1$; otherwise, $s_{m,n} = 0$. Each GU can only be associated with one UAV, thus, we have
\begin{equation}
\sum _{m\in \mathcal{M}} s_{m, n} = 1,\ \forall n \in \mathcal{N}.
\end{equation}
We adopt discrete transmit power control for UAVs, where the transmit power of UAV $m$ is selected from a predefined list as $P_{m} \in \{P_m^1,P_m^2,\dots,P_m^L\}$ with $P_m^1 < P_m^2 < \dots < P_m^L$. Let $p_{m}^{l} \in \{0,1\}$ with $l \in \mathcal{L} = \{1,2,\dots,L\}$ denote the power control variable, where UAV $m$ is assigned a transmit power at level $l$ (i.e., $P_m = P_m^l$) if $p_{m}^{l} = 1$; otherwise, $p_{m}^{l} = 0$. We note that only one power level can be assigned to the UAV, thus
\begin{equation}
\sum_{l \in \mathcal{L}} p_{m}^{l} \leq 1, \ \forall m \in \mathcal{M}.
\end{equation}

Let $x_{m,n}^{k,l}$ denote the signal to be transmitted to GU $n$ on channel $k$ from the UAV $m$ with transmit power $P_{m}^{l}$, which is expressed by
\begin{equation}
\begin{split}
\label{signal}
x_{m,n}^{k,l}= \beta_{m}^{k}s_{m,n}\sqrt{P_{m}^{l}}p_{m}^{l}t_{m,n}^{k},
\end{split}
\end{equation}
Each GU treats all the signals on the same sub-channel from other GUs connected to other UAVs as interference.
Thus, the received signal power at GU $n$ served by UAV $m$ over sub-channel $k$ and the transmit power $P_{m}^{l}$ and the co-channel interference caused by other UAVs operating on the same sub-channel $k$ with transmit power $P_{m'}^{l'}$ can be calculated as 
\begin{equation}
\begin{split}
\label{receivedsignal_S}
S_{m,n}^{k,l} = \left|g_{m,n}\right|^2 \beta_{m}^{k}P_{m}^{l}p_{m}^{l}
\end{split}
\end{equation}
\begin{equation}
\begin{split}
\label{Interference_I}
I_{m,n}^{k,l} = \sum\limits _{\substack{m'\in \mathcal{M} \\m'\neq m}}\sum\limits _{l'\in \mathcal{L}} \left | g_{m',n} \right |^2 \beta _{m'}^{k} P_{m'}^{l'}p_{m'}^{l'}.
\end{split}
\end{equation}
As a result, the signal-to-interference-plus-noise ratio (SINR) and the data rate of the GU $n$ served by the UAV $m$ over the sub-channel $k$ with transmit power $P_{m}^{l}$ can be expressed as
\begin{equation}
\begin{split}
\label{eq:sir}
\gamma _{m,n}^{k,l} =\frac{\left |g_{m,n}\right|^2 \beta^{k}_{m}P_{m}^{l}p_{m}^{l}}{I_{m,n}^{k,l} +\left |z_n\right |^2},
\end{split}
\end{equation}
\begin{equation}
\begin{split}
\label{rate_nmk}
R_{m,n}^{k,l}=s_{m,n}\log\left(1+\gamma_{m,n}^{k,l}\right).
\end{split}
\end{equation}
where $z_n \sim \mathcal{CN}\left (0,\sigma^2 \right )$ denotes the additive white Gaussian noise with zero mean at GU $n$.
In a network with multiple UAVs and GUs, the co-channel interference may be large enough to deteriorate the communication links of the UAVs.
Our goal is to improve the total achievable rate via user clustering and resource allocation in the system.
To achieve this, we focus on reducing interference by jointly optimizing clustering, sub-channel assignment, and transmit power allocation.

\subsection{Problem Formulation}
We aim to enhance the sum rate of the considered multi-UAV-aided wireless network by jointly optimizing clustering, sub-channel assignment, and UAV transmit power control. Let $\boldsymbol{\beta}=\{\beta^{k}_{m}\} \in \mathbb{R} ^{M \times K}$, $\boldsymbol{S}=\{s_{m,n}\} \in \mathbb{R} ^{M \times N}$, and $\boldsymbol{P} = \{p_{m}^{l}\} \in \mathbb{R}^{M\times L}$ be the sub-channel assignment, user association, and transmit power allocation vectors, respectively. The sum rate maximization problem is formulated as
\begin{equation}\label{eq:prob}
\begin{aligned}
\underset{\boldsymbol{\beta},\boldsymbol{S} ,\boldsymbol{P}}{\text{maximize}} \ 
& \sum\limits _{t\in \mathcal{T}}\sum\limits _{m\in \mathcal{M}}\sum\limits _{n\in \mathcal{N}}\sum\limits _{k\in \mathcal{K}} \sum\limits _{l\in \mathcal{L}} R_{m,n,t}^{k,l}\\
\text{s. t.} \quad
& \text{C1:} \ \sum _{k\in \mathcal{K}} \beta ^{k}_{m} \leq 1,\ \forall m\in \mathcal{M} ,\\
& \text{C2:} \ \beta ^{k}_{m} \in \{0,1\} ,\ \forall m \in \mathcal{M}, \forall k \in \mathcal{K} ,\\
& \text{C3:} \ \sum _{m\in \mathcal{M}} s_{m,n} = 1,\ \forall n\in \mathcal{N} ,\\
& \text{C4:} \ s_{m,n} \in \{0,1\} ,\ \forall m \in \mathcal{M}, \forall k \in \mathcal{K} ,\\
& \text{C5:} \ \sum _{l \in \mathcal{L}} p_{m}^{l} \leq 1,\ \forall m \in \mathcal{M} ,\\
& \text{C6:} \  p^{l}_{m} \in \{0,1\} ,\ \forall m \in \mathcal{M}, \forall l \in \mathcal{L}, 
\end{aligned}
\end{equation}
where C1 and C2 represent sub-channel assignment constraints; C3 and C4 represent user-association constraints; C5 and C6 are transmit power constraints. Problem \eqref{eq:prob} is defined as a combinatorial optimization problem that is difficult to solve optimally within a reasonable time (i.e., NP-hard), particularly in the case of large-scale networks. Moreover, it is necessary to transform the formulated problem into the QUBO or Ising model, so it can be solved by using a QAM. We decompose the original problem into two sub-problems, including user clustering and resource allocation. We then convert the problems into QUBO models and develop QA-based algorithms to solve the problems.
The following section presents the QA-based solution to the formulated problem.

\begin{remark}
    A server (central processing unit) is typically utilized to gather information about UAV positions and IoT devices' statuses, enabling effective management of network operations. Solving the joint optimization problem in a centralized manner aligns with emerging 5G/6G technologies, such as Software-Defined Networking (SDN) \cite{8802245} and Open Radio Access Network (ORAN) \cite{10601697}, which emphasize centralized control and flexibility. The proposed centralized implementation in our work is compatible with these technologies, leveraging centralized controllers for seamless coordination, efficient resource allocation, and global optimization. This integration ensures that the system operates cohesively to achieve the desired performance objectives while remaining adaptable to advancements in network architecture and infrastructure.
\end{remark}
\section{Quantum Annealing-Based Optimization Solution}
\label{sec_solution}
\renewcommand{\figurename}{{Fig.}}
\begin{figure}[t]
\centering
\includegraphics[width=\linewidth]{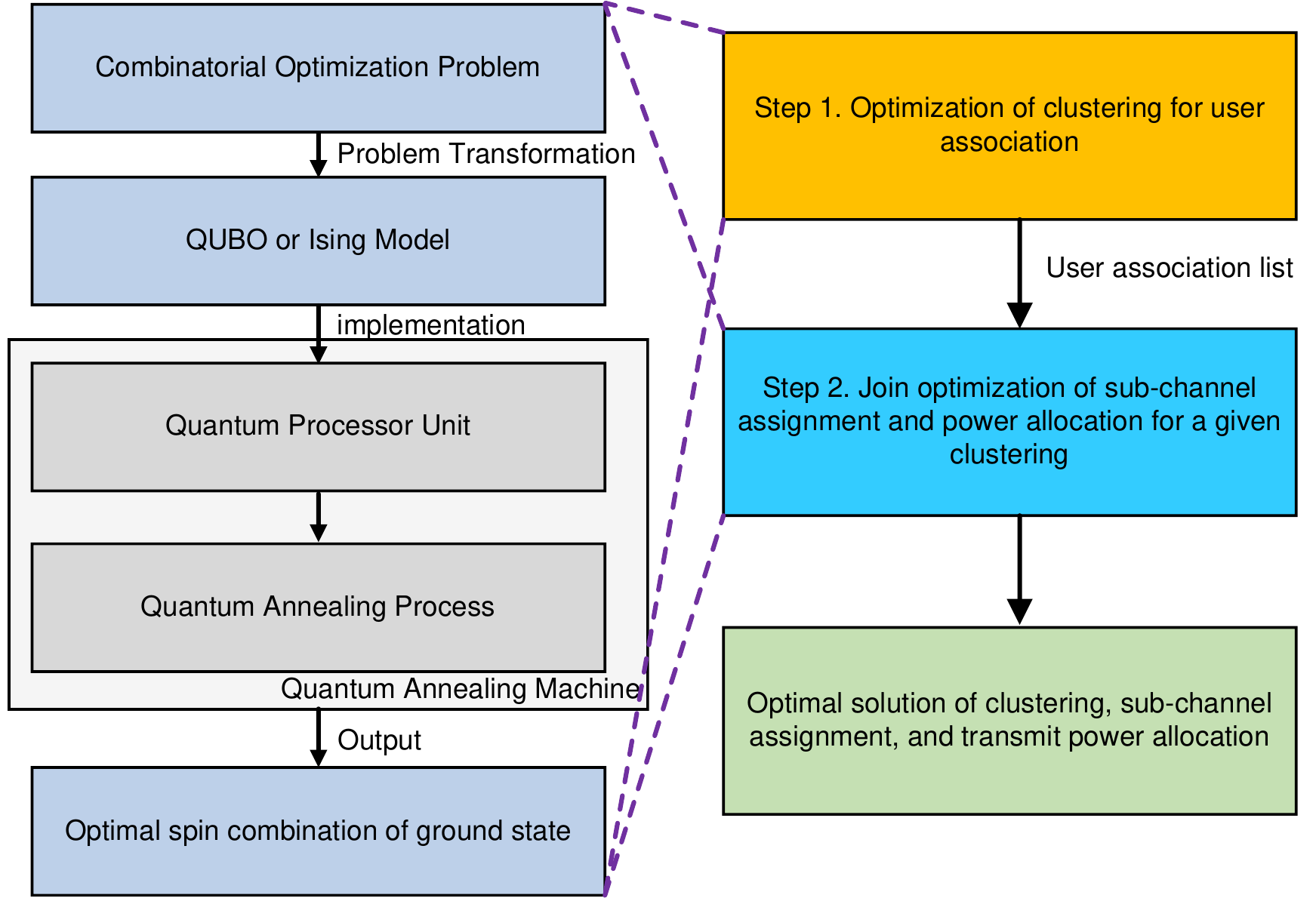}
\caption{The QA-based joint optimization problem of clustering, sub-channel assignment, and transmit power allocation.}
\label{fig:step}
\end{figure}
In this section, we describe how a QA-based optimization algorithm is used to solve the formulated problem. Firstly, we introduce the QUBO and Ising models. Secondly, we describe how to find the ground state of the energy function on the QAM. Then, we divide the problem into two sub-problems, as shown in Fig.~\ref{fig:step}: 1) the QA-based user clustering problem and 2) the joint optimization of sub-channel assignment and power allocation for a given clustering configuration. We obtain an approximately optimal solution for each optimization problem by solving these sub-problems in order, which are summarized below in Algorithm \ref{alg:PoEG2} and Algorithm \ref{alg:PoEG1}. Moreover, we convert each sub-problem into QUBO models, and we present a novel method to derive optimal scaling parameters to solve the QUBO models using a QAM.

\subsection{The QUBO and Ising Models}
A QAM uses the QUBO or Ising Hamiltonian of the Ising model to solve combinatorial optimization problems with a wide range of applications in computer science. The QUBO model can be expressed by the following optimization problem \cite{zaman2021pyqubo}.
\begin{equation}\label{eq:qubo_model}
\underset{\boldsymbol{x}}{\text{minimize}} \ H_{\text{QUBO}}(\boldsymbol{x}) = \boldsymbol{x} ^{T} \boldsymbol{Q} \boldsymbol{x},
\end{equation}
where $\boldsymbol{x}$ denotes a vector of binary decision variables, $\boldsymbol{x}^{T}$ represents the transpose of $\boldsymbol{x}$, and $\boldsymbol{Q}$ is an upper diagonal matrix.  
The QUBO model can also represent the binary combinatorial optimization problem with linear and quadratic terms. Considering the QUBO model with $N \times N$ binary vector in two dimensions, the QUBO model can be expressed as
\begin{equation}
\begin{split}
\label{QUBO}
H_{\text{QUBO}}(\boldsymbol{x})=\sum_{i=1}^{N}Q_{ii}x_i+\sum_{i=1}^{N}\sum_{j=1}^{N}Q_{ij}x_ix_j,
\end{split}
\end{equation}
where $x_{i} \in \{0,1\}$ is the $i$-th binary varible, $Q_{ii} \in \mathbb{R}$ and $Q_{ij}\in \mathbb{R}$ represent the linear and quadratic coefficients for $1 \leq j \leq i \leq N$, respectively. The QUBO model is computationally equivalent to the physical Ising model. 
By defining the spin decision variable $\sigma_{i} = 2x_{i}-1$, the QUBO model can be converted to the Ising model in the form of an energy function as follows \cite{zaman2021pyqubo}:
\begin{equation}
\begin{split}
\label{Ising}
H_{\text{Ising}}(\sigma)=\sum_{i=1}^{N}h_{i}\sigma_i+\sum_{i=1}^{N}\sum_{j=1}^{N}J_{ij}\sigma_i\sigma_j,
\end{split}
\end{equation}
where $h_{i}$ is the influence of the external magnetic field on $\sigma_i$, $J_{ij}$ is the interaction coefficient between $\sigma_i$ and $\sigma_j$ for $1 \leq j \leq i \leq N$. In Eq. \eqref{Ising}, the coefficients of the Ising Hamiltonian are given by $h_{i}=\frac{Q_{ii}}{2}+\sum_{j=1}^{N}\frac{Q_{ij}}{2} \ \forall i$ and $J_{ij}=\frac{Q_{ij}}{4} \ \forall i,j$. The QAM searches for the ground state of each spin $\sigma_{i}$ to minimize the energy function. 

To solve problem \eqref{eq:prob} using QAM, we aim to convert the objective function and constraints C1, C3, and C5 in problem \eqref{eq:prob} into the following QUBO form \cite{zaman2021pyqubo}:
\begin{equation}\label{eq:QUBO}
H_{\text{QUBO}} = H_{\text{cost}} + \lambda _{C} H_{\text{C}},
\end{equation}
where $H_{\text{cost}}$ and $H_{\text{C}}$ represent the objective function and penalty terms for the constraints in the original problem, and $\lambda_{C}$ denotes penalty parameters. 

\subsection{Quantum Annealing-Based Approach}
Quantum mechanical systems evolve over time according to the Schr\"{o}dinger equation \cite{doi:10.1126/science.1057726}
\begin{equation}
\begin{split}
\label{schr}
i\frac{d}{dt}|\psi(t)\rangle=H(t)|\psi(t)\rangle,
\end{split}
\end{equation}
where $|\psi(t)\rangle$ is the state vector of the quantum system, $H(t)$ is the time-dependent Hamiltonian, and $|\psi(0)\rangle$ is an initial state. The quantum system starts at $t=0$ in the ground state of $H(0)$, i.e., $|\psi_g(0)\rangle$. At time $T$, the state $|\psi(T)\rangle$ means the answer of the problem for $0\leq t \leq T$, where $T$ is the running time of the algorithm.

The quantum annealing process in Fig.~\ref{fig:step} makes the Hamiltonian evolve slowly enough, which is called an adiabatic process where the adiabatic quantum computation occurs \cite{RevModPhys.90.015002}. We take the advantage of the quantum adiabatic theorem considering a QAM of \textit{D-Wave}, for which the Hamiltonian is represented as 
\begin{equation}
\begin{split}\label{Dwave}
H(t) := \left(1-\frac{t}{T} \right)H(0) + \frac{t}{T}H_f, \ 0\leq t \leq T,
\end{split}
\end{equation}
where $t$ is a time instant and $T$ is the total period of the anneal process. On the right side of Eq. \eqref{Dwave}, the first term is the initial Hamiltonian, and the second term is the final Hamiltonian. As the system is annealed, the initial Hamiltonian decreases and the final Hamiltonian increases until $T$. At the end of the process, the Hamiltonian contains the only final Hamiltonian term. The total period $T$ should be large enough in order to satisfy the condition for adiabaticity \cite{RevModPhys.90.015002}. Then the final state of the qubits represents a low-energy solution. 
More specifically, both of two Hamiltonians are given as follows \cite{7782986}, respectively.
\begin{equation}
\begin{split}
\label{Hamil0}
H(0) := -\sum_{i=1}^{n}\sigma_i^x, \; |\psi_{init}\rangle := \frac{1}{\sqrt{2^{n}}}\sum_{z\in\{0,1\}}|z\rangle,
\end{split}
\end{equation}
\begin{equation}
\begin{split}
\label{Hamilf}
H_{f} := -\sum_{z\in\{0,1\}}^{n} f(z)|z\rangle\langle z|,
\end{split}
\end{equation}
where \eqref{Hamil0} is the initial Hamiltonian, \eqref{Hamilf} is the final Hamiltonian. The state vector of the quantum system evolves in a Hilbert space of dimension $2^{n}$. This $n$ qubit Hilbert space can be realized as a system of $n$ spin particles, where $|z_{i}=0\rangle$ corresponds to the $i$-th spin up in the $z$-direction and $|z_{i}=1\rangle$ corresponds to the $i$-th spin down in the $z$-direction. The adiabatic theorem ensures that the state $|\psi(T)\rangle$ at the end of the evolution has the ground state of $H_f$. Therefore, by the adiabatic theorem, we can take a list of qubit states corresponding to an eigen-state and the eigen-energy of the objective. Conventionally, the list of qubit states means the solution corresponding to the problem. Therefore, the list of qubit states is the solution of joint optimization of clustering, sub-channel assignment, and transmit power allocation. 

\subsection{QA-based User Clustering}\label{subsec:user_clustering}
In this subsection, we propose a QA-based clustering algorithm based on the $K$-means method \cite{selim1984k} to divide all GUs into $C$ clusters. In this work, we assume that each GU can only be associated with one BS in each time slot, and each UAV serves a group of GUs that are located in proximity of each other. As a result, we can effectively group GUs into $C$ clusters by using the $K$-means clustering method owing to low implementation complexity. We also note that such a solution has been widely adopted in previous studies \cite{8700188, 9121255, AHMAD2007503}.
However, the clustering problem has been proven NP-hard \cite{mahajan2012planar}, which makes it difficult to obtain the optimal solution. To tackle the NP-hard issue of the problem, we propose a QA-based clustering algorithm, which yields the optimal solution to the clustering problem.
Similar to \cite{9121255}, we assume that the location of each UAV is fixed at the center of each cluster and that the number of clusters is the same as the number of UAVs. By using the user association variable $s_{m,n}$ and the distance between GUs and UAVs $d_{m,n}$, the clustering problem is formulated as follows:

\begin{equation}\label{eq:prob3}
\begin{aligned}
\underset{\boldsymbol{S}}{\text{minimize}} \ 
& \sum\limits _{m\in \mathcal{M}}\sum\limits _{n\in \mathcal{N}} s_{m,n}d_{m,n}\\
\text{s. t.} \quad
& \text{C3:} \ \sum _{m\in \mathcal{M}} s_{m,n} = 1,\ \forall n\in \mathcal{N} ,\\
& \text{C4:} \ s_{m,n} \in \{0,1\} ,\ \forall m \in \mathcal{M}, \forall n \in \mathcal{N} .
\end{aligned}
\end{equation}

To solve the problem using quantum annealing, we first define a new binary variable $X_{m,n}$, which indicates that GU $n$ is associated with the UAV $m$ if $X_{m,n} = 1$; otherwise, $X_{m,n} = 0$.
In addition, each GU can only be associated with one UAV, constraint C3 can be transformed to 
\begin{equation}
    \sum _{m \in \mathcal{M}} X_{m,n} \leq 1,\ \forall n\in \mathcal{N}.
\end{equation}
Let $\boldsymbol{X} = \{X_{m,n}\} \in {\{0,1\}}^{M\times N}$ denote the user association vector. Thus, problem \eqref{eq:prob3} can be reformulated as
\begin{equation}\label{eq:refo3}
\begin{aligned}
\underset{\boldsymbol{X}}{\text{minimize}} \ 
&\sum\limits _{m\in \mathcal{M}}\sum\limits _{n\in \mathcal{N}} X_{m,n}d_{m,n}\\
\text{s. t.} \quad
& \text{C7:} \ \sum _{m \in \mathcal{M}} X_{m,n} \leq 1,\ \forall n\in \mathcal{N} ,\\
& \text{C8:} \ X_{m,n} \in \{0,1\} ,\ \forall m \in \mathcal{M}, n \in \mathcal{N}. 
\end{aligned}
\end{equation}
Problem \eqref{eq:refo3} is also defined as a combinatorial optimization problem. Therefore, we can transform the problem \eqref{eq:refo3} into the following QUBO model:
\begin{equation}\label{eq:QUBOH10}
H_{\text{QUBO}} = H_{cost} + \lambda _{p} H_{\text{C7}},
\end{equation}
where $H_{cost}=\sum\limits _{m\in \mathcal{M}}\sum\limits _{n\in \mathcal{N}}X_{m,n}d_{m,n}$, and $H_{\text{C7}}$ represents the penalty term for constraint C7, which is defined as follows:
\begin{equation}\label{Const10b}
H_{\text{C7}} = \sum\limits _{n\in \mathcal{N}}\left(\sum\limits _{m\in \mathcal{M}} X_{m,n}-1\right)^2.
\end{equation}
In this model, $\lambda_{p}$ denotes a penalty factor that is used to adjust the degree of influence.


Inserting the QUBO model into the QAM allows us to implement the model using the Constrained Quadratic Model (CQM) provided by \textit{D-Wave}. This implementation takes place in the Quantum Processing Unit (QPU) to match the QUBO model to the physical system accurately. The QAM ultimately produces the optimal spin combination to minimize the inserted QUBO model. Adjusting the penalty factor is necessary to obtain the global solution with a feasible set of the QUBO model. Notably, the penalty factor is automatically adjusted by the CQM solver during the implementation process, eliminating the need for manual adjustments. However, when implementing the QUBO model in a conventional QPU solver that does not utilize CQM, the penalty factor must be adjusted. This adjustment process is iteratively repeated several times. Eventually, the QAM produces the true ground state with a feasible set of the QUBO model, resulting in the penalty term becoming zero.


In the following, we present a suggested range for the penalty factor $\lambda_{p}$. It's important to note that this factor is not unique, and there are multiple values that can be effectively used \cite{https://doi.org/10.48550/arxiv.1811.11538}. If the penalty factor is excessively large, the penalty term may overpower the objective function information. Conversely, if the penalty factor is too small, searching for feasible solutions becomes challenging. Therefore, striking a balance within the suggested range is crucial for achieving effective results.

\begin{lemma}
    Given the penalty term $H_{\text{C7}}$ from \eqref{eq:QUBOH10}, and assuming $H_{\text{C7}} \leq (M-1)^2N$, with $H'_{cost}=\sum\limits _{m\in \mathcal{M}}\sum\limits _{n\in \mathcal{N}}X'_{m,n}d_{m,n} \ \forall X'_{m,n} \neq X_{m,n}$ and $X^*$ denotes the optimal solution while $\mathcal{G}$ denotes the infeasible solution space, we obtain appropriate penalty factor $\lambda_{p}$ in the following range:
    \begin{equation}\label{eq:lambdapppppp}
\lambda_{p} \in \left[\underset{X_{m,n} \in \mathcal{G}}{\text{max}}\left(\frac{H_{cost}'-H_{cost}}{H_{C7}}\right)\\,\underset{X_{m,n}\in \mathcal{G}}{\text{max}} H_{cost}\right].
\end{equation}
\end{lemma}
\begin{proof}
    From \eqref{eq:QUBOH10}, the penalty term $H_{\text{C7}}$ increases according to the degree of constraint violation as follows:
\begin{equation}\label{eq:lambdapp}
H_{QUBO}({X}^*) < H_{cost} +\lambda_{p}H_{C7}, \ \forall X_{m,n} \in \mathcal{G},
\end{equation}
Therefore, we can obtain a valid lower bound for the penalty factor as
\begin{equation}
    \begin{aligned}
        \lambda_{p} &> \underset{X_{m,n}}{\text{max}}\frac{H_{QUBO}(X^*)-H_{cost}}{H_{C7}} \\
        & > \underset{X_{m,n}}{\text{max}}\left(\frac{H_{cost}'-H_{cost}}{H_{C7}}\right),
    \end{aligned}
\end{equation}
Furthermore, a relevant study \cite{VERMA2022100594} provides insights into determining the upper bound of $\lambda_{p}$.
Thus, we can determine the range for $\lambda_{p}$ as given by equation \ref{eq:lambdapppppp}
\end{proof}

The computational complexity associated with obtaining the optimal range for the penalty factor is $\mathcal{O}(2^{MN})$. Ultimately, we can employ a local search method that utilizes an f-flip neighborhood to determine the penalty factor \cite{VERMA2022100594}. The proposed QA-based user clustering method is described in Algorithm \ref{alg:PoEG2}.

\begin{algorithm}[t]
\label{alg:PoEG2}
    \caption{The QA-Based Clustering Algorithm with CQM in \textit{D-Wave}}
    \SetAlgoLined
    \SetKwInOut{Input}{Input}
    \SetKwInOut{Output}{Output}
    
    \Input{Locations of GUs and UAVs.}
    
    \Output{The user association vector}
    Build the CQM object
    
    Define the binary variable $X_{m,n}$        
    \For  {$m \in \mathcal{M}$} {
        \For {$n \in \mathcal{N}$}{
            $H_{cost} \gets X_{m,n}d_{m,n}$  
        }
    }       

    $H_{QUBO} \gets H_{cost} + \lambda_{p}H_{C7}$

    Run the CQM sampler
    
    Return the optimal spin combination $X_{m,n}^{*}$   
    
\end{algorithm}

\subsection{Joint Optimization of Sub-channel Assignment and Power Allocation for a Given Clustering Configuration}
In this subsection, we present the QA-based sub-channel assignment and power allocation algorithm under the given clustering configuration. Most of the previous works adopt classical optimization methods, such as successive convex approximation \cite{10013035} and Karush-Kuhn-Tucker \cite{7873307}, to assign the sub-channels and allocate the transmit power to the UAVs with low complexity. However, using such methods to obtain the optimal solutions for non-convex problems takes time and effort. In the following, we propose a novel method that leverages the QA algorithm for sub-channel assignment and power allocation. This approach enables us to achieve the optimal solution for the optimization problem.

The problem is formulated as follows:
\begin{equation}\label{eq:prob4}
\begin{aligned}
\underset{\boldsymbol{\beta},\boldsymbol{P}}{\text{maximize}} \ 
& \sum\limits _{m\in \mathcal{M}}\sum\limits _{n\in \mathcal{N}}\sum\limits _{k\in \mathcal{K}} \sum\limits _{l\in \mathcal{L}} R_{m,n}^{k,l}\\
\text{s. t.} \quad
& \text{C1:} \ \sum _{k\in \mathcal{K}} \beta ^{k}_{m} \leq 1,\ \forall m\in \mathcal{M} ,\\
& \text{C2:} \ \beta ^{k}_{m} \in \{0,1\} ,\ \forall m \in \mathcal{M}, \forall k \in \mathcal{K} ,\\
& \text{C5:} \ \sum _{l \in \mathcal{L}} p_{m}^{l} \leq 1,\ \forall m \in \mathcal{M} ,\\
& \text{C6:} \  p^{l}_{m} \in \{0,1\} ,\ \forall m \in \mathcal{M}, \forall l \in \mathcal{L}. 
\end{aligned}
\end{equation}

To solve this problem, we first define a new binary variable $X_{m,k,l}$, which indicates that UAV $m$ is associated with channel $k$ and transmit power $P_m^l$ if $X_{m,k,l} = 1$; otherwise, $X_{m,k,l}=0$. The SINR for the communication link between UAV $m$ and GU $n$ on sub-channel $k$ with transmit power level $l$ can be rewritten as
\begin{equation}\label{eq:sir_new}
\gamma_{m,n}^{k,l}=\frac {\left |g_{m,n}\right |^2 P_{m}^{l} X_{m,k,l}}{\displaystyle \sum\limits_{\substack{m'\in \mathcal{M} \\m'\neq m}} \sum\limits_{l' \in \mathcal{L}} \left |g_{m',n}\right |^2 P_{m'}^{l'} X_{m',k,l'}+\left |Z_0\right |^2}.
\end{equation}
We then introduce Kronecker delta $\delta_{m,m'}$ and $\delta_{k,k'}$, where $\delta_{i,j}$ is a piecewise function of two variables $i$ and $j$, i.e., $\delta_{i,j} = 1$ if $i=j$, and $\delta_{i,j} = 0$ otherwise. Thus, the received signal power $S$ and interference $I$ can be recast as
\begin{equation}
\begin{split}\label{rece_S}
\text{S} =& s_{m,n} \left|g_{m,n}\right|^2 P_{m}^{l} X_{m,k,l},
\end{split}
\end{equation}
\begin{equation}
\begin{split}\label{rece_I}
\text{I} =& \sum\limits _{m'}\sum\limits _{k'} \sum\limits _{l'}s_{m,n}\left|g_{m',n}\right|^2 P_{m'}^{l'} X_{m',k',l'}\delta_{k,k'}(1-\delta_{m,m'}),
\end{split}
\end{equation}
where $\delta_{m,m'}$ and $\delta_{k,k'}$ denote the interference occurred to GU $n$ from other UAVs sharing the same sub-channel $k$. 

Since each UAV can only be assigned one sub-channel and one power level, constraints C1 and C5 can be transformed to the following constraint:
\begin{equation}
    \sum _{k \in \mathcal{K}} \sum _{l \in \mathcal{L}} X_{m,k,l} \leq 1,\ \forall m\in \mathcal{M}.
\end{equation}

Please note that the objective function of problem \eqref{eq:prob4} is differentiable and can be transformed into polynomials using Taylor series approximation. However, the QUBO model is a binary combinatorial optimization problem, which make it challenging to directly translate the objective into a QUBO -formulation. 
{To address this issue, we use the inequality $\log(1+x) \leq x$.} Subsequently, we transform the approximated problem into the QUBO model. For problem \eqref{eq:prob4}, it is approximated as the following problem in the vicinity of the zero-point.
\begin{equation}\label{eq:prob5}
\begin{aligned}
\underset{\boldsymbol{\beta},\boldsymbol{P}}{\text{maximize}} \ 
& \sum\limits _{m\in \mathcal{M}}\sum\limits _{n\in \mathcal{N}}\sum\limits _{k\in \mathcal{K}} \sum\limits _{l\in \mathcal{L}} \gamma _{m,n}^{k,l}\\
\text{s. t.} \quad
& \text{C1:} \ \sum _{k\in \mathcal{K}} \beta ^{k}_{m} \leq 1,\ \forall m\in \mathcal{M} ,\\
& \text{C2:} \ \beta ^{k}_{m} \in \{0,1\} ,\ \forall m \in \mathcal{M}, \forall k \in \mathcal{K} ,\\
& \text{C5:} \ \sum _{l \in \mathcal{L}} p_{m}^{l} \leq 1,\ \forall m \in \mathcal{M} ,\\
& \text{C6:} \  p^{l}_{m} \in \{0,1\} ,\ \forall m \in \mathcal{M}, \forall l \in \mathcal{L}. 
\end{aligned}
\end{equation}
Nevertheless, it is important to note that the optimal solution to the original problem \eqref{eq:prob4} may differ from the optimal solution to the approximated problem \eqref{eq:prob5}. In order to obtain the same optimal solution for both problems, we propose to use the MILFP method, which is detailed in subsection \ref{MILFPsection}. The MILFP approach is capable of determining the optimal solution around the zero-point.

To make problem \eqref{eq:prob5} solvable by using QAM, the next step is to convert the problem into a QUBO model. By denoting $\boldsymbol{X} = \{X_{m,k,l}\} \in {\{0,1\}}^{M\times K\times L}$ as the sub-channel assignment and power allocation vector, problem \eqref{eq:prob5} can be reformulated as
\begin{equation}\label{eq:refo1}
\begin{aligned}
\underset{\boldsymbol{X}}{\text{minimize}} \ 
&-\sum\limits _{m\in \mathcal{M}}\sum\limits _{n\in \mathcal{N}}\sum\limits _{k\in \mathcal{K}} \sum\limits _{l\in \mathcal{L}} \frac {\text{S}}{\text{I}+\left|Z_0\right|^2}\\
\text{s. t.} \quad
& \text{C9:} \ \sum _{k \in \mathcal{K}} \sum _{l \in \mathcal{L}} X_{m,k,l} \leq 1,\ \forall m\in \mathcal{M} ,\\
& \text{C10:} \ X_{m,k,l} \in \{0,1\} ,\\
& \quad \quad \forall m \in \mathcal{M}, k \in \mathcal{K} , l \in \mathcal{L}. 
\end{aligned}
\end{equation}


It is worth noting that the noise power $\left|Z_{0}\right|^2$ in the objective function is a constant and not affected by any control variables. Therefore, we focus on minimizing the $-\frac{S}{I}$ ratio rather than the  $-\frac{S}{I+|Z_0|^2}$ without losing generality. Thus, problem \eqref{eq:refo1} is equivalently recast as

\begin{equation}\label{eq:refo16}
\begin{aligned}
\underset{\boldsymbol{X}}{\text{minimize}} \ 
&-\sum\limits _{m\in \mathcal{M}}\sum\limits _{n\in \mathcal{N}}\sum\limits _{k\in \mathcal{K}} \sum\limits _{l\in \mathcal{L}}\text{S} +\lambda _{\text{I}}\sum\limits _{n\in \mathcal{N}}\text{I}\\
\text{s. t.} \quad
& \quad \text{C9, C10}. 
\end{aligned}
\end{equation}
Consequently, the original problem can be converted into the following QUBO model:
\begin{equation}\label{eq:QUBOH}
H_{\text{QUBO}} = H_{cost2} + \lambda _{p2} H_{\text{C9}},
\end{equation}
where $H_{cost2} = -\sum\limits _{m\in \mathcal{M}}\sum\limits _{n\in \mathcal{N}}\sum\limits _{k\in \mathcal{K}} \sum\limits _{l\in \mathcal{L}}\text{S} +\lambda _{\text{I}}\sum\limits _{n\in \mathcal{N}}\text{I}$, and $H_{\text{C9}}$ represents the penalty term for constraint C9, which is defined as follows:
\begin{equation}\label{Const9b}
H_{\text{C9}} = \sum\limits _{m\in \mathcal{M}}\left(\sum\limits _{k\in \mathcal{K}} \sum\limits _{l\in \mathcal{L}}X_{m,k,l}-1\right)^2.
\end{equation}
In this model, $\lambda_{\text{I}}$ is a scaling parameter that can be derived by a parametric algorithm for MILFP. The parameter $\lambda_{p2}$ denotes the penalty factor, which adjusts the degree of influence. Again, to effectively solve the problem using a QAM, it is essential to set appropriate values for $\lambda_{\text{I}}$ and $\lambda_{p}$. Hence, in the following subsection, we present the proposed methods to derive $\lambda_{\text{I}}$ and $\lambda_{p}$.

\subsection{Deriving The Optimal Scaling Parameter and Appropriate Penalty Factor}\label{MILFPsection}
By setting suitable values for $\lambda_{\text{I}}$ and $\lambda_{p}$ in the QAM of \textit{D-Wave}, we can utilize quantum annealing to efficiently search for the lowest energy states, which correspond to optimal or near-optimal solutions for the problem. To derive the scaling parameter, we employ a parametric algorithm for the MILFP problem within the QUBO model \eqref{eq:QUBOH}.

The general form of MILFP can be stated as follows \cite{https://doi.org/10.1002/aic.14185}:
\begin{equation}
\begin{split}
\label{MILFP}
\text{maximize}\left \{Q(x) = \frac{N(x)}{D(x)} \bigg| x \in \mathcal{F} \right \},
\end{split}
\end{equation}
where variables $x$ can be both continuous and discrete, $\mathcal{F}$ is the feasible set, and the denominator function $D(x)$ is always positive, i.e., $D(x) > 0 \ \forall x \in \mathcal{F}$. The functions of the numerator $N(x)$ and denominator $D(x)$ can be linear or nonlinear. We can rewrite the above equation as follows:
\begin{equation}
\begin{split}
\label{MILFP2}
F(q) = \text{maximize}\left \{N(x) - q\cdot D(x)| x \in \mathcal{F} \right \},
\end{split}
\end{equation}
where $q$ is a variable parameter. Note that the optimal solution of the parametric objective function $F(q)$ has only one zero-point \cite{6858622}, which is the same as its global optimal solution. Therefore, the parametric element is derived as $q^{*} = \frac{N(x^{*})}{D(x^{*})} = \text{max} \left \{ \frac{N(x)}{D(x)}| x \in \mathcal{F}\right \}$ if only and if $F(q^{*})=F(q^{*},x^{*})=\text{max}\left \{(N(x)-q^{*}D(x)| x \in \mathcal{F} \right \}=0$ where $x^{*}$ is the global optimal solution. 

By comparing \eqref{MILFP2} and \eqref{eq:QUBOH}, the signal power $S$ and interference $I$ correspond to $N(x)$ and $D(x)$, respectively.
Hence, the optimal scaling parameter can be derived as 
\begin{equation}\label{eq:QUBOH2}
\lambda_{\text{I}} = \frac{\left|g_{m,n}\right|^2 P_{m}^{l}X_{m,k,l}^*}{\sum\limits _{m'}\sum\limits _{k'} \sum\limits _{l'}\left|g_{m',n}\right|^2 P_{m'}^{l'}\delta_{k,k'}(1-\delta_{m,m'})X_{m',k,l'}^*},
\end{equation}
where $X^*$ denotes the optimal solution of the QUBO model.

Using this approach, we can calculate the optimal scaling parameter for the QUBO model.
However, it is not possible to directly set this scaling parameter in the QUBO model because the solution formulation process for the QAM in \textit{D-Wave} is iterative, involving updates and repeatitions. Therefore, we divide $\lambda_{I}$ into a numerator $\lambda_{num}$ and denominator $\lambda_{den}$ to ensure adherence to the QUBO model. The process of calculating the scaling parameter is shown in Algorithm \ref{alg:PoEG1}.
 
\begin{algorithm}[t!]
\label{alg:PoEG1}
    \caption{QA-Based Sub-channel Assignment and Power Allocation}
    \SetAlgoLined
    \SetKwInOut{Input}{Input}
    \SetKwInOut{Output}{Output}
    \DontPrintSemicolon
    \Input{Locations of GUs and UAVs, user association vector.}
    
    \Output{Sub-channel assignment and power allocation decisions.}
     
    Define the binary variable $X_{m,k,l}$ for QAM.

    \textbf{Repeat}
    
    Adjust the penalty factor $\lambda_{p2}$.
        
    \For  {$m \in \mathcal{M}$} {
        \For    {$n \in \mathcal{N}$} {
            \For    {$k \in \mathcal{K}$}{
                \For    {$l \in \mathcal{L}$}{

        Set the temp variable $T \gets 0$

        Set the $\lambda_{den} \gets 0$

        \For    {$m' \in \mathcal{M}$} {

            \If     {$m \neq m'$}{

                \For{$l' \in \mathcal{L}$}{

                    $T \gets |g_{m,n}|^{2}P_{m'}^{l'}X_{m',k,l'}X_{m,k,l}$    

                    \If{$X_{m',k,l'}X_{m,k,l}$}{
                        $\lambda_{den} \gets |g_{m',n}|^{2}P_{m'}^{l'}$
                    }       
                            }
                        }
                    }
        \eIf { $ \lambda_{den} = 0$}{
            
            $H_{cost2} \gets -|g_{m,n}|^2P_{m}^{l}X_{m,k,l}s_{m,n}$
            
        }{
            $\lambda_{num} \gets |g_{m,n}|^2P_{m}^{l}$

            $H_{cost2} \gets -|g_{m,n}|^{2}P_{m}^{l}X_{m,k,l}s_{m,n} + s_{m,n}T\frac{\lambda_{num}}{\lambda_{den}}$
        }
                }
            }
        }
    }        

    $H_{QUBO} \gets H_{cost2} + \lambda_{p2}H_{C9}$

    \textbf{Until} {QAM output the feasible solution}
    
    Return optimal spin combination $X_{m,k,l}^{*}$
    
\end{algorithm}

\begin{lemma}
    Given the penalty term $H_{\text{C9}}$ and $H_{\text{C9}} \leq (KL-1)^2M$ according to \eqref{Const9b} and $X^*$ denotes the optimal solution while $\mathcal{G}$ denotes the infeasible solution space, we obtain appropriate penalty factor $\lambda_{p}$ in the following range:
\begin{equation}\label{eq:lambdappppp2}
\lambda_{p2} \in \left[\underset{X_{m,k,l} \in \mathcal{G}}{\text{max}}\left(\frac{H_{cost2}'-H_{cost2}}{H_{C9}}\right)\\,\underset{X_{m,k,l}\in \mathcal{G}}{\text{max}} H_{cost2}\right]
\end{equation}
\end{lemma}
\begin{proof}
Regarding the penalty factor $\lambda_{p2}$ in \eqref{eq:QUBOH}, we apply the same method as in subsection \ref{subsec:user_clustering} to determine the appropriate range for this parameter. Specifically, the penalty term $H_{\text{C9}}$ in \eqref{eq:QUBOH} increases according to the degree of constraint violation as follows:
\begin{equation}\label{eq:lambdapppp}
H_{QUBO}({X}^*) < H_{cost2} +\lambda_{p2}H_{C9}, \ \forall X_{m,k,l} \in \mathcal{G},
\end{equation}
Therefore, we can obtain a valid lower bound for the penalty factor as
\begin{equation}
    \begin{aligned}
        \lambda_{p2} & > \underset{X_{m,k,l}}{\text{max}}\frac{H_{QUBO}(X^*)-H_{cost2}}{H_{C9}} \\
        & > \underset{X_{m,k,l}}{\text{max}}\left(\frac{H_{cost2}'-H_{cost2}}{H_{C9}}\right),
    \end{aligned}
\end{equation}
 Finally, an f-flip neighborhood-based search method is used to determine the penalty factor \cite{VERMA2022100594}.
Thus, a valid penalty factor can be chosen as \ref{eq:lambdappppp2}
\end{proof}
The computational complexity associated with obtaining the optimal range for penalty factor $\lambda_{p2}$ is $\mathcal{O}(2^{MKL})$.  Algorithm \ref{alg:PoEG1} outlines the process of QA-based joint optimization of sub-channel assignment and power allocation.

\section{Simulation Results}\label{Simulationresult}
In this section, we provide the simulation results to demonstrate the effectiveness of the proposed QA-based algorithms in a multi-UAV wireless network. More specifically, GUs are assigned to designated UAVs using the QA-based clustering algorithm. The sub-channels and transmit power are determined through the QA-based joint optimization of sub-channel assignment and power allocation scheme.
We assume that the UAVs are uniformly distributed as a regular polygon within a three-dimensional area of $2.5\text{km} \times 2.5\text{km} \times 100 \text{m}$. The preplacement configurations of the UAVs are shown in Fig.~\ref{fig:UAVcase}. 
\renewcommand{\figurename}{{Fig.}}
\begin{figure}[t]
\centering
\includegraphics[width=\linewidth]{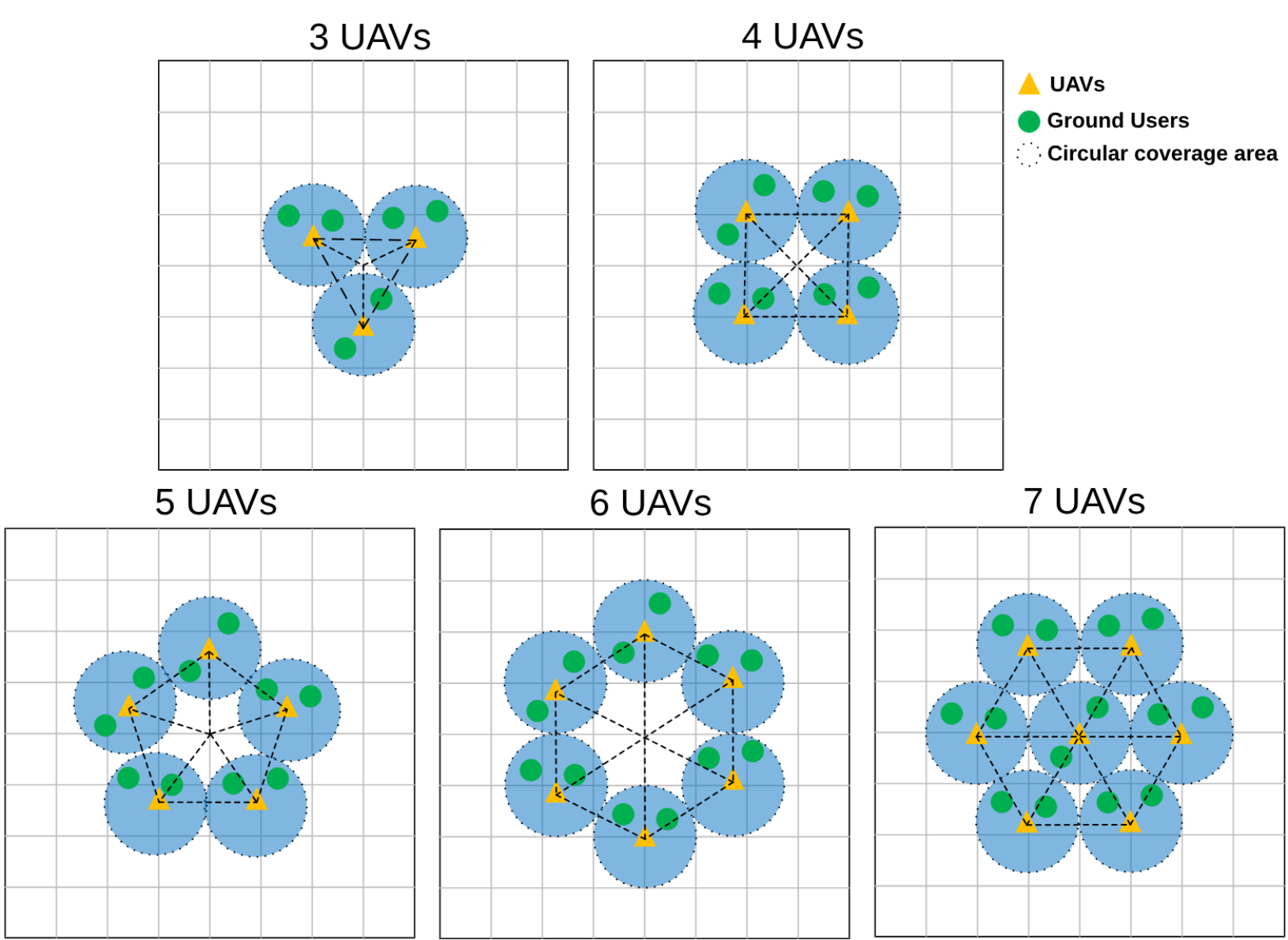}
\caption{Preplacement configuration of the UAVs in a two-dimensional space.}
\label{fig:UAVcase}
\end{figure}
Each UAV has a circular coverage area with a radius of $500\text{m}$ and can serve up to 30 GUs. The GUs are randomly and uniformly located in the network's service area. The simulation utilizes air-to-ground propagation parameters specific to a dense urban environment \cite{6863654}.
To ensure reliable results, the algorithms are run multiple times, and the outcomes are averaged. The simulation parameters are detailed in Table~\ref{table1}.
\begin{table}[t]
	\centering 
	\caption{Simulation Parameters}  
	\label{table1} 
	\begin{tabular}{| c | c|} 
	\hline 
	\textbf{Parameter} & \textbf{Value}\\ 
	\hline
	Carrier frequency & $2$~GHz\\  
	\hline
	Number of GUs & $100$\\ 
	\hline
	Number of UAVs $M$& $\{1, 2, 3, 4, 5, 6, 7\}$\\ 
	\hline
	Number of sub-channels $K$ & $\{2, 3\}$\\
	\hline
	Service area radius, $r_c$ & $500$ m\\
	\hline
        Channel parameter $a$ & $9.6$\\ 
	\hline
        Channel parameter $b$ & $0.16$\\ 
	\hline
	Channel parameter $\eta^{LoS}$ & $1$ dB\\ 
	\hline
        Channel parameter $\eta^{NLoS}$ & $20$ dB\\ 
	\hline
	Transmit power $P_m$ & $\{10, 15, 20, 25, 30\}$ dBm\\ 
        \hline
	Noise power ${|Z_0|}^2$ & $-96$ dBm\\
	\hline
	UAVs' altitude & $100$ m\\

	\hline 
 	QUBO solver & \textit{D-Wave} hybrid solver \cite{dwave}\\
	\hline
	\end{tabular}
\end{table}

To demonstrate the effectiveness of the proposed QA algorithm in large-scale wireless communication networks, we consider two simulation scenarios: 1) only increasing the number of UAVs, and 2) increasing the numbers of the sub-channels and UAVs. We compare our proposed QA algorithm with two benchmarks, namely the Steepest Descent (SD) \cite{10.1162/neco.1992.4.2.141} and Simulated Annealing (SA) \cite{doi:10.1126/science.220.4598.671}.
The SD algorithm, a discrete analogue of gradient descent, determines the best move by using local minimization instead of computing a gradient. It falls under the category of heuristic methods, which may result in suboptimal solutions for certain problems. However, SD has the advantage of often finding local optimal solutions faster than exhaustive search methods. 
The SA algorithm is a metaheuristic approach that approximates global optimization in large search spaces. Inspired by the physical process of annealing, which involves gradually cooling a high-temperature material to achieve an optimal configuration, SA starts with a high temperature and progressively cools down over time. It relies on random search and probabilistic acceptance of new solutions to gradually converge to the optimal solution. Moreover, we also compare the clustering performance of these schemes with the $K$-means++ algorithm \cite{10.5555/1283383.1283494}. These algorithms are implemented by using Python 3.7 on a computer with an NVIDIA RTX 3080 Ti GPU and we run the quantum annealing algorithm on \textit{D-Wave} quantum annealing machines (QAMs).
\subsection{Complexity Analysis}
In this subsection, we analyze the complexity of the proposed algorithm. As stated before, the sum rate maximization problem \eqref{eq:prob} has been divided into two subproblems, which are then transformed into QUBO models. We propose a QA-based algorithm to solve these subproblems sequentially. The computational complexity of the algorithms primarily depends on the CQM solver, which consists of a classical front-end working in conjunction with a quantum back-end. The front-end takes inputs, including an optional time limit $T$, and initiates a set of heuristic solvers running on classical CPUs and GPUs to search for high-quality solutions. Each heuristic classical solver contains a quantum module that formulates and sends quantum queries to a \textit{D-Wave} QPU. The QPU's responses to these queries may be used to improve the quality of a current set of solutions. This combined approach of classical and quantum solution methods working together is known as hybrid acceleration. As a result, defining a precise time complexity is challenging. While this hybrid method is not solely reliant on quantum techniques, the hybrid workflow can achieve better solutions more quickly compared to a purely classical workflow. 

To simplify the analysis of time complexity, we focus on the computation performed by the \textit{D-wave} quantum computer. The time required to execute a single quantum machine instruction on a QPU, known as the qubit-processing-unit access time (QPU access time), consists of the programming time, annealing time, readout time, and delay time \cite{dwave}. {The annealing time is particularly relevant to the time complexity of the problem $(t_f)$, while the remaining factors are more related to overhead considerations. In \cite{Mukherjee_2015}, an approximate evaluation of the time complexity was proposed based on the probability of overcoming the energy barrier of the QA algorithm. Therefore, the time complexity of the proposed QA algorithm can be estimated as $t_f \approx \mathcal{O}(e^{\sqrt{N}})$ where $N$ is the number of qubits used to represent the variables of the optimization problem. Similarly, the authors in \cite{Mukherjee_2015} estimated the time complexity of the Simulated annealing (SA) as $ \mathcal{O}(e^{N})$. In fact, the running time of the SA algorithm also depends on the number of iterations, the number of state changes per loop (temperature), and other factors. 
For classical algorithms such as K-means++ and Steepest Descent (SD), we also provide evaluations of the time complexity of these algorithms. Specifically: for the K-means++ algorithm, the time complexity can be estimated as 
$\mathcal{O}(NTK)$; for the SD algorithm,  $t_f \approx \mathcal{O}(NGD)  $, and for the SA algorithm, $t_f \approx \mathcal{O}(e^{N})$,  where $T$ is the number of iterations until convergence, 
 $K$ is the number of clusters, and 
 $G$ is the time to calculate the gradient. For the SD method, the number of iterations depends on specific conditions of the algorithm, such as the gradient of the objective function, learning rate, and required precision, so the time complexity of this algorithm varies depending on the actual problem. However, it should be noted that the complexity of the proposed algorithm is being considered for execution on a quantum computer, which is different from the classical algorithms, and the efficiency of the proposed algorithm is demonstrated in the simulation results provided below.  }

\subsection{Performance Analysis}
In this subsection, we analyze the performance of the proposed algorithm.
\renewcommand{\figurename}{{Fig.}}
\begin{figure}[t]%
\subfigure[Solution vs. Energy.]{{\includegraphics[width=0.5\textwidth]{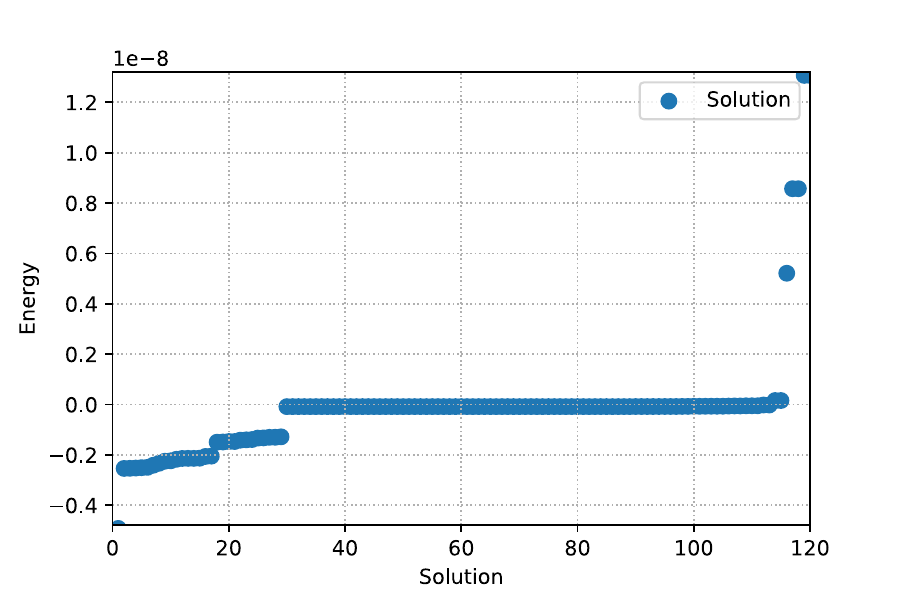} }}%
\hfill
\subfigure[Feasible solution vs. Energy.]{{\includegraphics[width=0.5\textwidth]{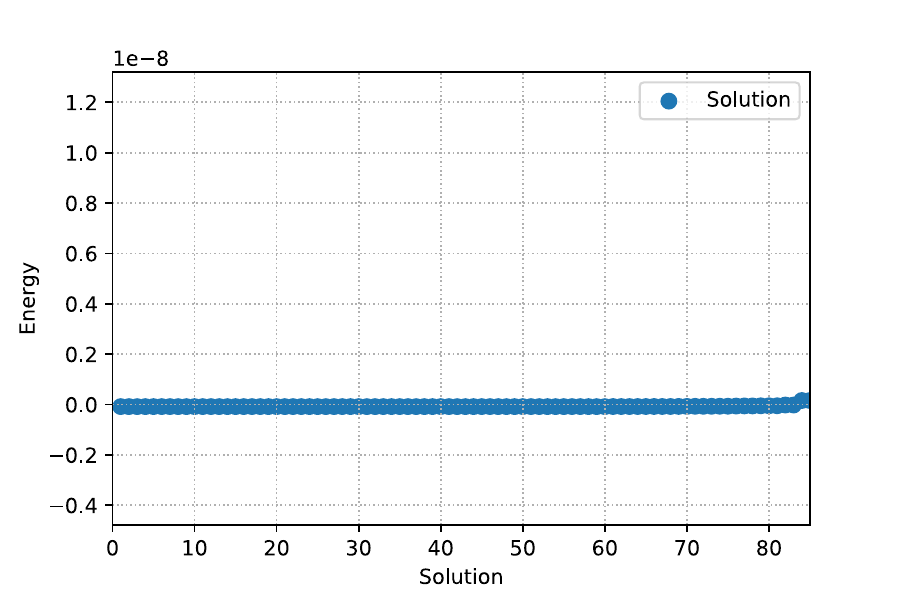} }}%
\caption{The solution and energy from the QAM in the case of scenario 2 and the number of UAV is 7.}%
\label{fig.Energy}%
\end{figure}
Fig.~\ref{fig.Energy} shows an example of the outcomes of the QAM for the proposed QA algorithm. 
As shown in Fig.~\ref{fig.Energy}(a), the solutions from the QAM include the feasible and infeasible solutions owing to the various number of sampling. In our formulated problem \eqref{eq:prob}, there exist constraints that induce the infeasible solution if it can not meet the constraints. In addition, we used linear approximation around zero-point to approximate like problem \eqref{eq:prob5}. To achieve this, we applied the parametric algorithm for the MILFP problem in the QUBO model. Therefore, the energy of the QUBO model should be $0$ by the MILFP definition. Fig.~\ref{fig.Energy}(b) shows the feasible solution to problem \eqref{eq:refo16} while considering the constraints. We can observe that the solutions locate around the energy level of $0$. In the case of problem \eqref{eq:QUBOH}, if there remains the penalty term $\lambda_{p2}H_{C9}$ or the optimal scaling parameter $\lambda_I$ is not utilized, the energy of QUBO model can not be $0$. This result confirms that the condition of linear approximation is satisfied, and our proposed QA algorithm provides solutions that closely approximate the global optimum. Among the various feasible solutions, we select the solution with the lowest energy value. 

To demonstrate the effectiveness of the proposed QA algorithm in solving the clustering problem, we compare its performance with the SD, SA, and K-means++ algorithms. For the purpose of comparison, we utilize the objective function of problem \eqref{eq:prob3} as the evaluation metric. The simulation results are depicted in Fig.~\ref{fig.Clustering}, and a detailed breakdown is provided in Table~\ref{table2}. The GUs from different clusters are indicated by different colors. Each UAV is marked with a triangle of a distinct color. The circular coverage area of each UAV is represented by a dotted line circle of the corresponding color. 

In Fig.~\ref{fig.Clustering}(a), the SD algorithm yields only a locally optimal solution, resulting in 29$\%$ poor matching GUs associated with UAVs. In Fig.~\ref{fig.Clustering}(b), the SA algorithm performs better than the SD algorithm, but it requires more running time. In Fig.~\ref{fig.Clustering}(c), the K-mean++ algorithm provides an appropriate and effective solution, with only 6$\%$ poor matching GUs located at the edge of the UAVs' service areas. The running time for this result is 0.18$s$. However, the proposed QA algorithm ensures a global optimal solution, as shown in Fig.~\ref{fig.Clustering}(d). Despite the presence of GUs at the edge of the service areas, our proposed QA algorithm can produce a near-optimal solution by solving the distance-based combinatorial optimization problem. In terms of running time, the K-means++ algorithm exhibits the best performance among the classical algorithms. However, the QPU access time of the proposed QA algorithm is only $0.032$ seconds. It is note worthy that the numerical results are presented using Min-Max Normalization with the objective. The SD and SA algorithms do not provide effective solutions. 
\renewcommand{\figurename}{{Fig.}}
\begin{figure*}[t]
\begin{center}
\hspace{2mm}
\subfigure[SD-based clustering]{{\includegraphics[width=0.22\textwidth]{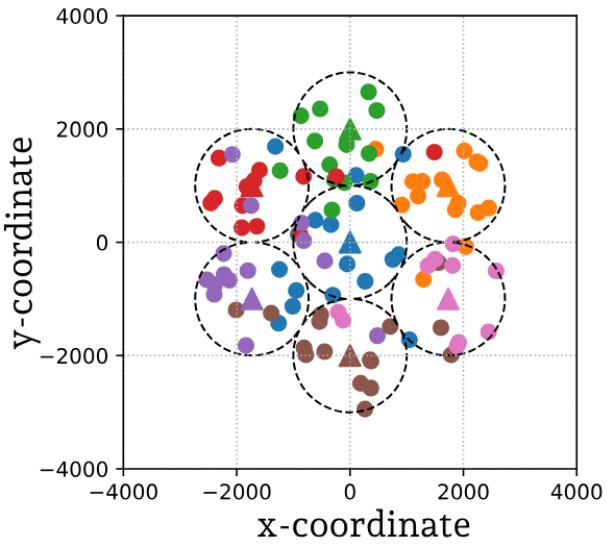} }}%
\hspace{-1mm}
\subfigure[SA-based clustering]{{\includegraphics[width=0.22\textwidth]{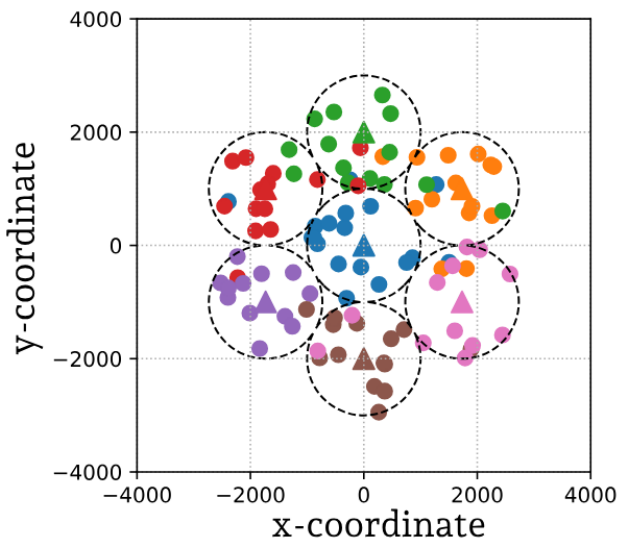} }}%
\hspace{-1mm}
\subfigure[K-means++ clustering]{{\includegraphics[width=0.22\textwidth]{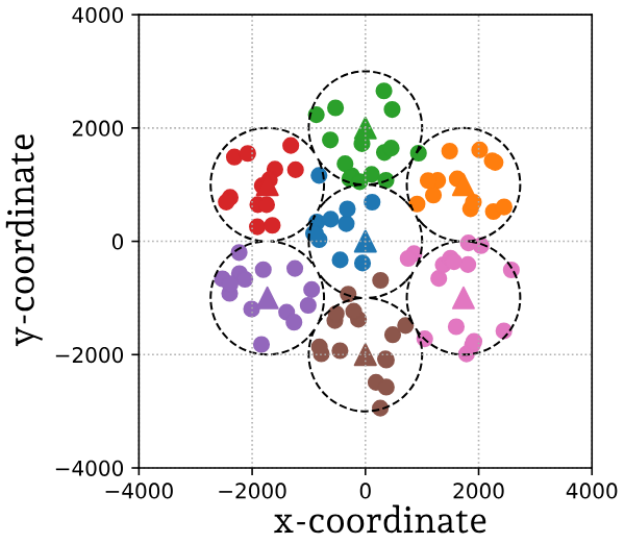} }}%
\hspace{-1mm}
\subfigure[QA-based clustering]{{\includegraphics[width=0.27\textwidth]{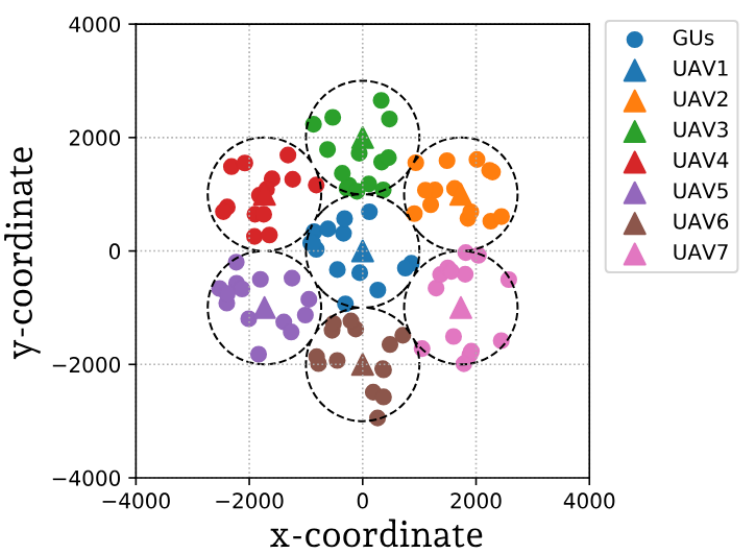} }}%
\end{center}
\caption{Clustering comparison of SD, SA, K-means++, and QA.}%
\label{fig.Clustering}%
\end{figure*}
\renewcommand{\figurename}{{Fig.}}
\begin{figure*}[t]
\begin{center}
\hspace{2mm}
\subfigure[5 UAVs (K-means++)]{{\includegraphics[width=0.22\textwidth]{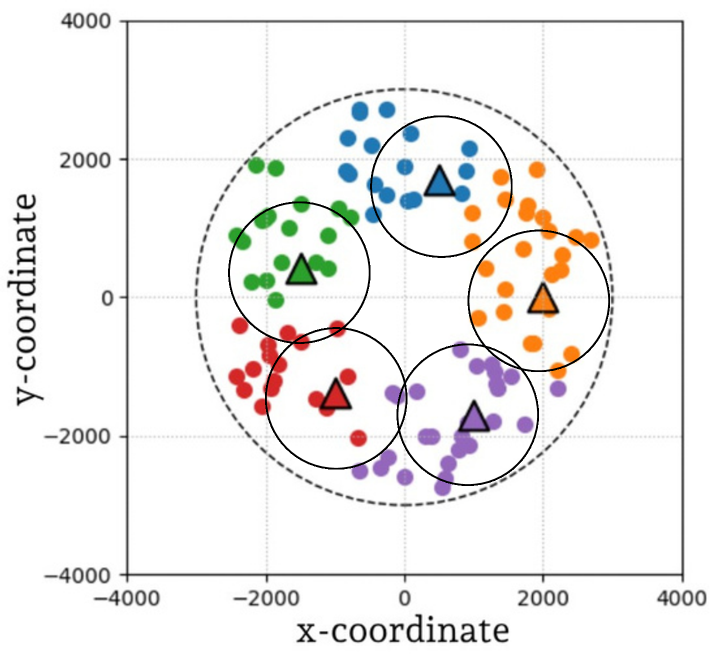} }}%
\hspace{-1mm}
\subfigure[5 UAVs (QA)]{{\includegraphics[width=0.22\textwidth]{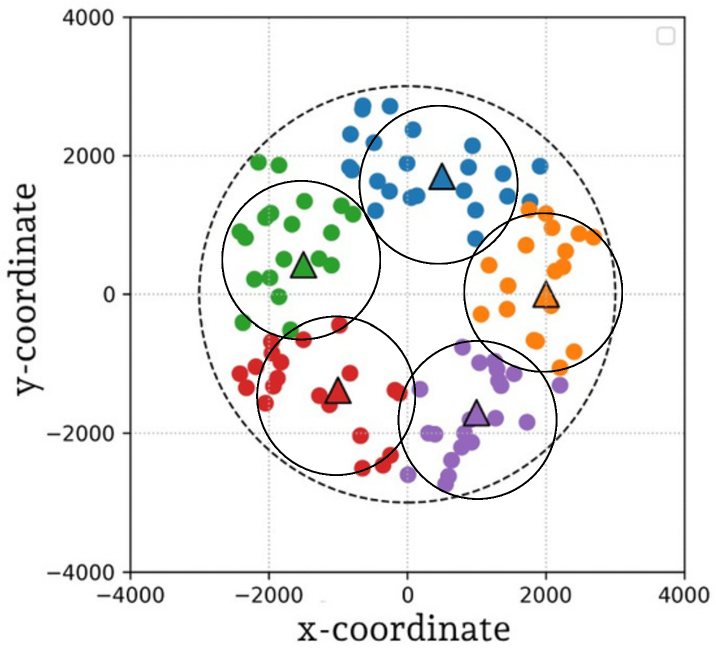} }}%
\hspace{-1mm}
\subfigure[6 UAVs (K-means++)]{{\includegraphics[width=0.22\textwidth]{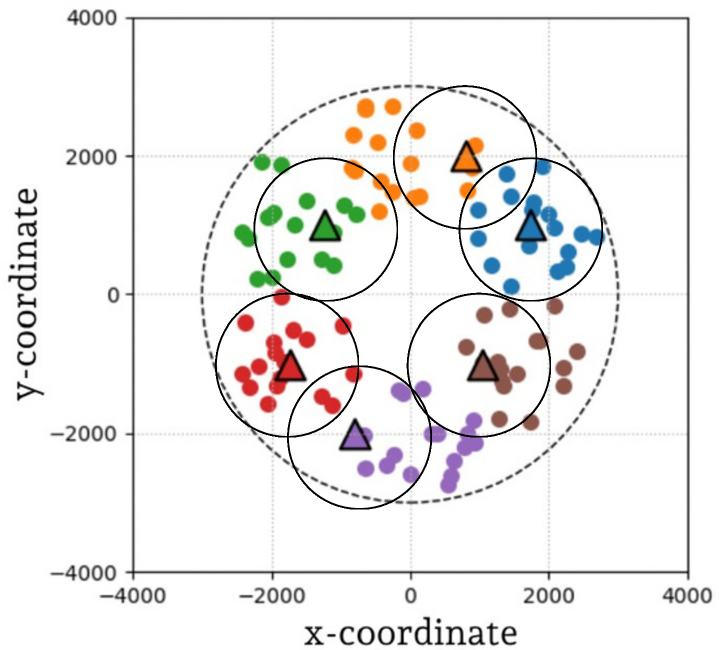} }}%
\hspace{-1mm}
\subfigure[6 UAVs (QA)]{{\includegraphics[width=0.22\textwidth]{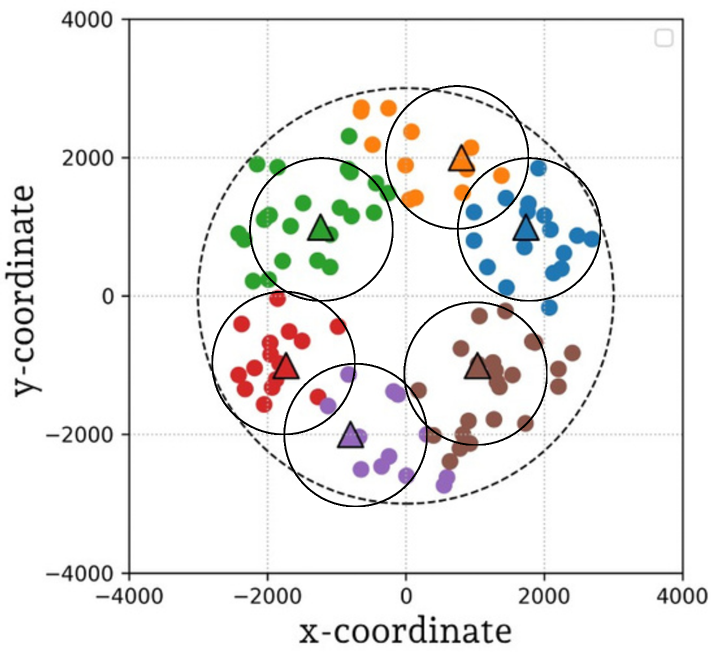} }}%
\end{center}
\caption{{Clustering comparison of K-means++ and QA in case the coordinates of the UAVs change.}}%
\label{fig.Compare_Clustering_Kmean_QA}%
\end{figure*}
\begin{table}
\centering 
\caption{Clustering Performance of each algorithm}  
\label{table2} 
\begin{tabular}{|c|c|c|c|c|}
\hline
\begin{tabular}[c]{@{}c@{}}\end{tabular} & SD & SA  & K-means++ & QA  \\ \hline
Poor matching GUs ($\%$)                                                      
& $29\%$ & $16\%$  & $6\%$  & $0\%$ \\ \hline
Normalized value                                                      
& 1 & 0.624  & 0.065  & 0 \\ \hline
Running time (s)
& 5.78 & 72.76 & 0.18 & 0.032 \\ \hline
\end{tabular}
\end{table}
{In Fig.~\ref{fig.Compare_Clustering_Kmean_QA}, we also perform the changes of the coordinates of the UAVs in two scenarios when the number of UAVs is $5$ and $6$ and compare with the K-means++ algorithm. From Fig.~\ref{fig.Compare_Clustering_Kmean_QA}(a) and Fig.~\ref{fig.Compare_Clustering_Kmean_QA}(c), when the UAV changes coordinates, the K-means++ algorithm gives suboptimal performance results 
when clustered GUs are not suitable for UAVs. This led to a number of UAVs
being overloaded as they have to manage too many GUs while the remaining UAVs handleless. Furthermore, there are GUs that connect to distant UAVs, leading to poor channel conditions ad reduced system performance. Meanwhile, Fig~\ref{fig.Compare_Clustering_Kmean_QA}(b) and Fig.~\ref{fig.Compare_Clustering_Kmean_QA}(d) shows the effectiveness of QA algorithms on user clustering problems.
The QA algorithm produces a good clustering ability when most GUs are connected to the nearest UAVs, leading to improved system performance compared to the K-means++ algorithm. It is evident that our proposed QA algorithm is better suited for the clustering problem.}
\renewcommand{\figurename}{{Fig.}}
\begin{figure}[t]%
\subfigure{{\includegraphics[width=0.5\textwidth]{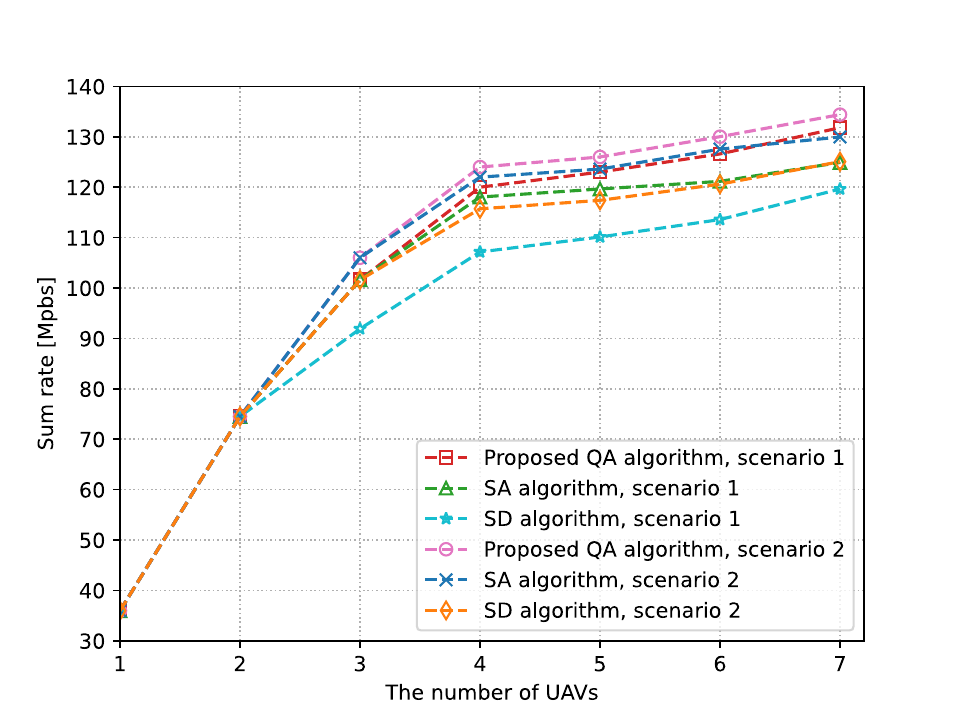} }}%
\caption{{Sum rate of the proposed QA algorithm and benchmarks with a varying number of UAVs}}%
\label{fig.sumrate}%
\end{figure}

Fig.~\ref{fig.sumrate} depicts the sum rate achieved by different algorithms for varying numbers of UAVs in each scenario. The overall sum rate increases as the number of UAVs increases. This is because more GUs can be served, even with a limited number of channels. We can observe that scenario 2 exhibits a better performance than scenario 1 because the UAVs in scenario 2 can utilize more sub-channels to serve GUs. Additionally, Fig.~\ref{fig.sumrate} shows that our proposed QA algorithm maximizes the sum rate more effectively than the other algorithms. Meanwhile, the SD algorithm can only provide a locally optimal solution in these scenarios. Notably, the difference in sum rate between the proposed QA and SA algorithms becomes apparent starting from $4$ UAVs. This discrepancy arises from the impact of the simulation size on the performance of the SA algorithm. Even if there are more channels in scenario 2, the SA algorithm allocates resources to the GUs with weaker channels and transmit power, leading to poor service. 

Fig.~\ref{sumrate_gus} depicts the total data rate achieved by different algorithms when the number of GUs is varied and the number of UAVs is 5 in scenario 1. It is observed that as the number of GUs increases, the sum rate increases. However, the proposed QA algorithm in different GUs cases gives a much better sum rate than the SD or SA algorithms. From these results, we can see that our proposed QA algorithm provides a better solution compared to other benchmarks. Fig.~\ref{fig:runningtime} illustrates the average running time of the QA, SD, and SA algorithms. It shows that the running time of all algorithms, except the proposed QA algorithm, increases as the simulation size grows. The SD algorithm can obatin the local optimal point faster than the SA algorithm due to its superior time complexity. The proposed QA algorithm, on the other hand, consistently provides the results within a very short time. As an example, the QPU access time of the proposed QA algorithm is merely $0.032$ seconds. In addition, the running time remains the same across all cases. This uniformity arises from the dependence of the QA running time on the energy gap. In this work, the problem is solved using the D-wave hybrid solver, which consists of a classical CPU and a QPU. However, specific parameters related to running time control are not provided, making it difficult to determine the total running time, which consists of both CPU and QPU processing time. Even if there is no clear total running time of the proposed QA algorithm, we focus on QPU access time because it can be implemented using only the QPU. Based on these findings, we can verify that the proposed QA algorithm outperforms other benchmarks in terms of running time. 

\renewcommand{\figurename}{{Fig.}}
\begin{figure}[t]
\centering
\includegraphics[width=\linewidth]{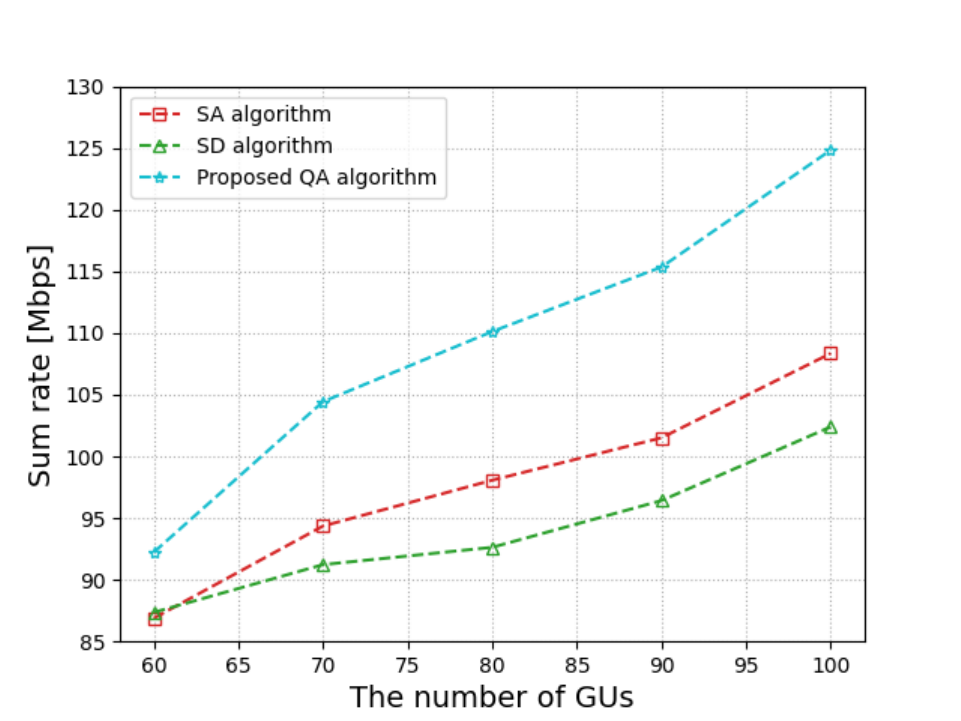}
\caption{{Sum rate of the proposed QA algorithm and benchmarks with a varying number of GUs.}}
\label{sumrate_gus}
\end{figure}

\subsection{Practicality and Challenges}
The practicality of the proposed framework hinges on the availability and capabilities of quantum machines used within a centralized network controller. The number of qubits required is dictated by the complexity of the optimization and processing tasks. For example, solving the formulated problem~\eqref{eq:refo1} may necessitate $M\times K\times L$ logical qubits, which could surpass the capacity of current quantum systems, such as D-Wave's sixth-generation processor supporting up to $7000$ qubits \cite{mcgeoch2022d}. Moreover, implementing a QUBO model on D-Wave's QAM requires careful consideration of the relationship between logical and physical qubits. Logical qubits represent the variables in the problem, while physical qubits are the actual qubits on the QPU that realize these logical qubits. Due to the QPU's specific topology, a process called minor embedding is employed to represent the logical qubits~\cite{PhysRevA.92.042310}.
While hybrid workflows like D-Wave's Hybrid Solver Service (HSS) can mitigate this by reducing the required number of qubits, the inherent limitations of quantum hardware—such as noise, error rates, and scalability—remain significant challenges. Furthermore, current quantum machines are largely research-focused and accessible only via online platforms, with commercially viable solutions still in development. Despite these issues, our framework provides a promising foundation for future research and practical advances as quantum technologies mature and become more accessible. This work underscores the potential of quantum-enabled approaches in addressing complex optimization problems in wireless communication and related domains.

\section{Conclusion}\label{sec_conclusion}
This study proposes a novel QA-based approach for addressing the user clustering and resource allocation problem in multi-UAV-aided wireless networks.
We formulate the sum rate maximization problem as a combinatorial optimization problem. To solve this problem efficiently, we decompose it into two sub-problems: 1) QA-based clustering and 2) joint optimization of sub-channel assignment and power allocation given the clustering. Each step involves formulating a combinatorial optimization problem and converting it into a QUBO model, which can be rapidly solved by using a QAM with CQM solver. In the second step, we also present a MILFP-based technique to derive the optimal scaling parameter. It is noteworthy that the MILFP theory satisfies various transformation challenges, such as linear approximation and fractional function. Simulation results demonstrate the significant effectiveness of the proposed algorithm over other benchmarks in terms of clustering, sum rate, and running time. Especially, Fig.~\ref{fig.Energy}(b) exhibits a feasible set of the MILFP-based QUBO formulation, emphasizing the requirement for the energy to be zero. Our future work will focus on applying this method to address diverse optimization problems in large-scale wireless communication systems. We will explore strategies for converting hard-constraint problems to unconstrained problems, considering factors such as the bound of SINR or throughput in wireless systems. 
\renewcommand{\figurename}{{Fig.}}
\begin{figure}[t]
\centering
\includegraphics[width=\linewidth]{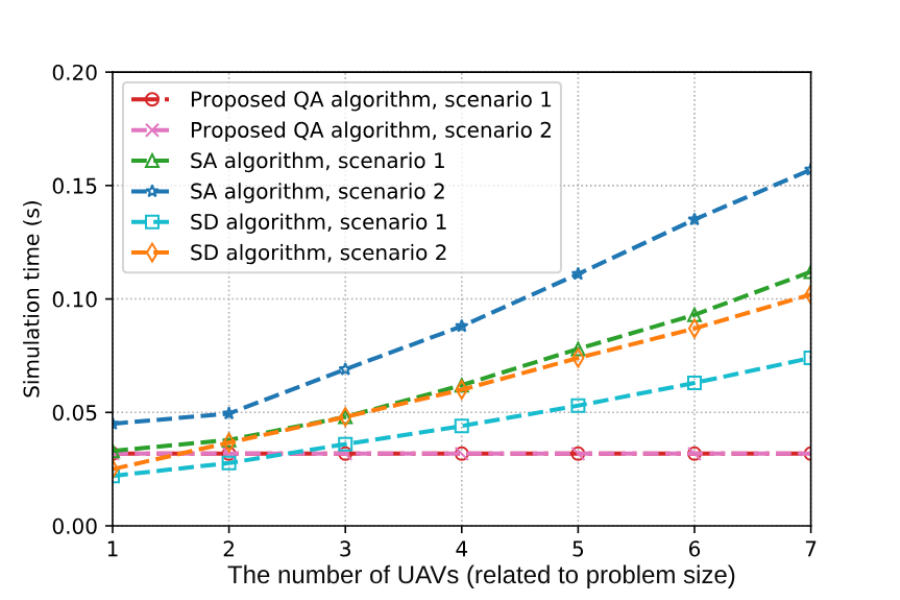}
\caption{The average running times (in seconds) for each algorithm for a varying number of UAVs.}
\label{fig:runningtime}
\end{figure}

\bibliographystyle{IEEEtran}
\bibliography{bibliography}

\begin{IEEEbiography}[{\includegraphics[width=1in,height=1.25in,clip,keepaspectratio]{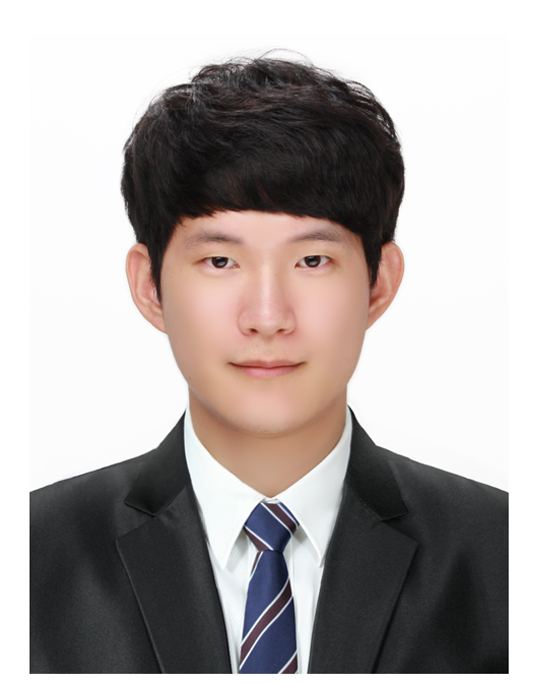}}]{Seon-Geun Jeong} received a B.S. degree in electrical engineering from Pusan National University, Busan, South Korea, in 2017, where he is currently pursuing an integrated Ph.D. program at the Department of Information Convergence Engineering, Artificial Intelligence Convergence Research Center, Pusan National University, Busan, South Korea.
His current research interests include quantum annealing, quantum machine learning, and quantum information.
\end{IEEEbiography}

\begin{IEEEbiography}[{\includegraphics[width=1in,height=1.25in,clip,keepaspectratio]{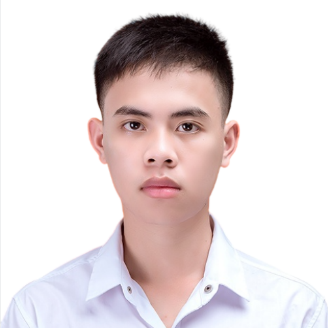}}]{Pham Dang Anh Duc} is currently pursuing a B.S. degree in Computer Engineering at the School of Information and Communication Technology, Hanoi University of Science and Technology, Hanoi, Vietnam. His research interests include edge computing, quantum computing, and optimization for wireless communication.
\end{IEEEbiography}

\begin{IEEEbiography}[{\includegraphics[width=1in,height=1.25in,clip,keepaspectratio]{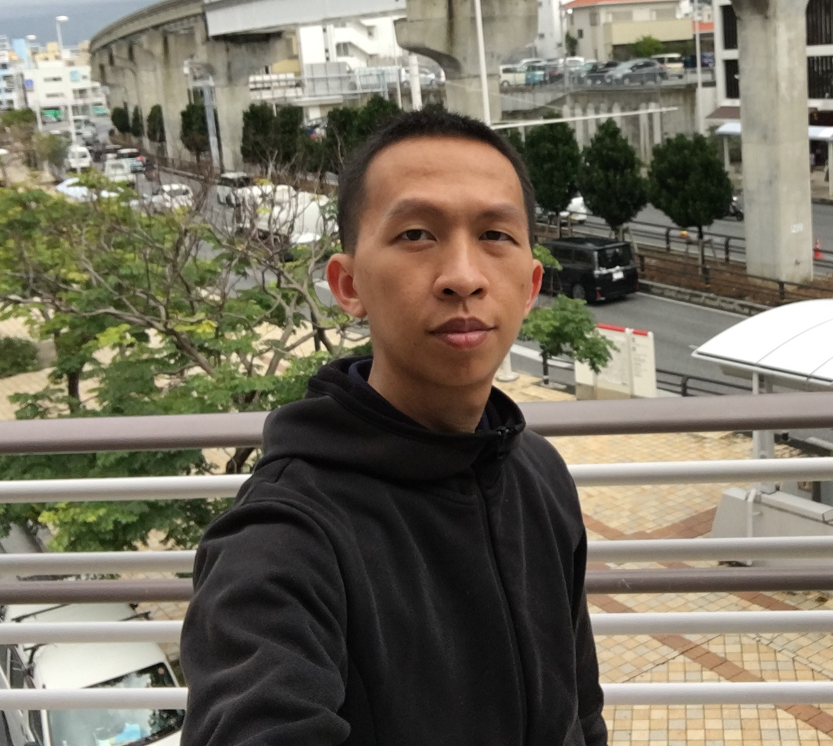}}]{Quang Vinh Do} received the bachelor’s degree in electrical engineering from Ho Chi Minh City University of Technology, Vietnam, in 2009, the master’s degree in electronic and computer engineering from RMIT University, Australia, in 2013, and the Ph.D. degree in electrical engineering from the University of Ulsan, South Korea, in 2020. He was a Postdoctoral Researcher with the University of Ulsan, from September 2020 to February 2021. From March 2021 to August 2023, he was a Post-Doctoral Research Fellow with the Artificial Intelligence Research Center, Pusan National University, South Korea. He is currently a Lecturer at Ton Duc Thang University, Vietnam. His research interests include developing and applying artificial intelligence techniques to wireless communication networks.
\end{IEEEbiography}

\begin{IEEEbiography}[{\includegraphics[width=1in,height=1.25in,clip,keepaspectratio]{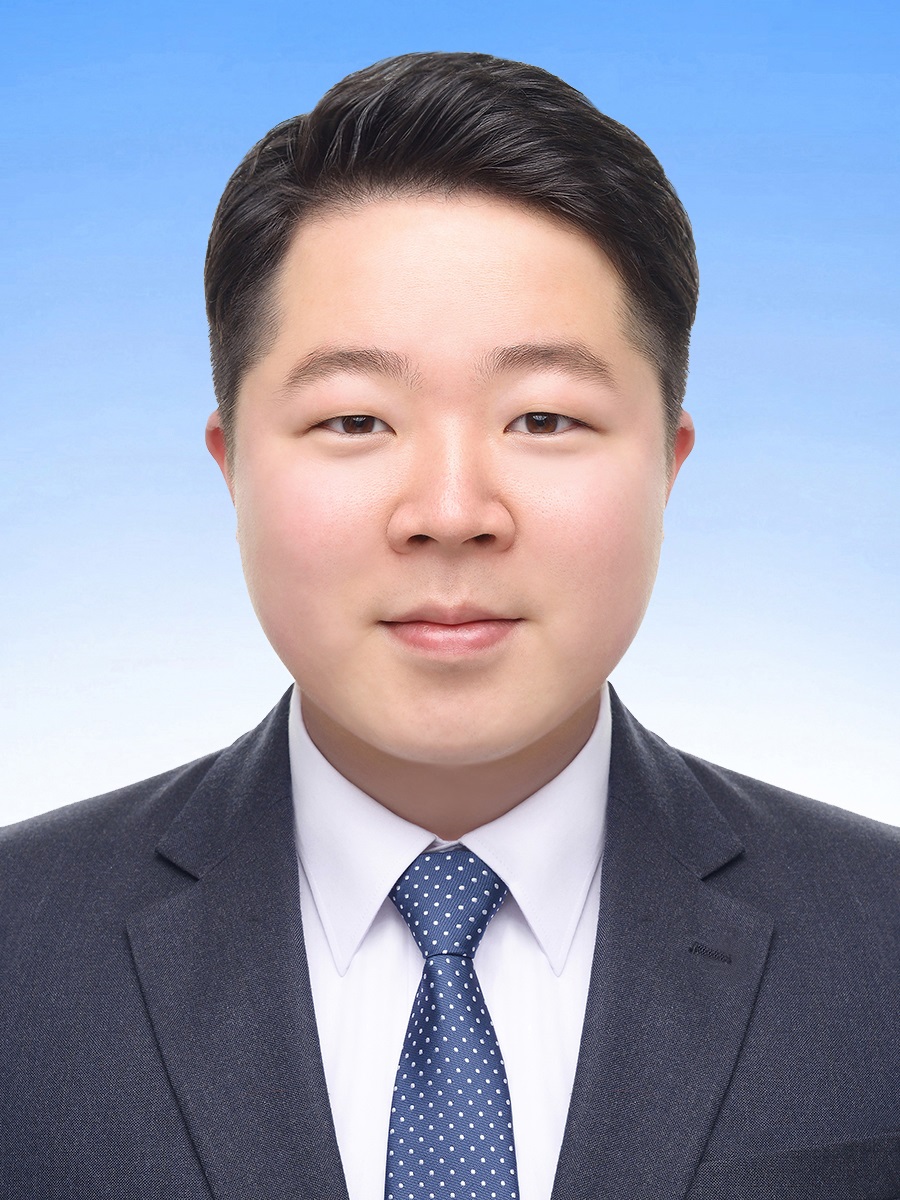}}]{Dae-Il Noh} received a double B.S. degree in Electrical Engineering and Biomaterials Science and an M.S. degree from the Department of Information Convergence Engineering, Pusan National University, Busan, Korea, in 2018 and 2022, respectively. He is currently pursuing a Ph.D. degree at the Department of Information Convergence Engineering, Pusan National University, Busan, Korea. His research interests include signal processing, machine learning, deep learning, wireless communication, and quantum-inspired evolutionary computation.
\end{IEEEbiography}

\begin{IEEEbiography}[{\includegraphics[width=1in,height=1.25in,clip,keepaspectratio]{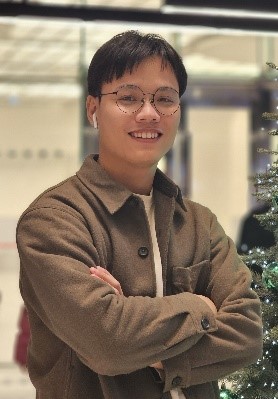}}]{Nguyen Xuan Tung} received the B.S. degree in electronics and telecommunications engineering from the Hanoi University of Science and Technology, Vietnam, in 2019, and the Master degree in intelligent mechatronic engineering from Sejong University, Korea, in 2022. He is currently working toward a PhD degree with the Department of Information Convergence Engineering at Pusan National University, South Korea. His research interests include optimization theory, machine learning, and quantum computing for wireless communications.
\end{IEEEbiography}

\begin{IEEEbiography}[{\includegraphics[width=1in,height=1.25in,clip,keepaspectratio]{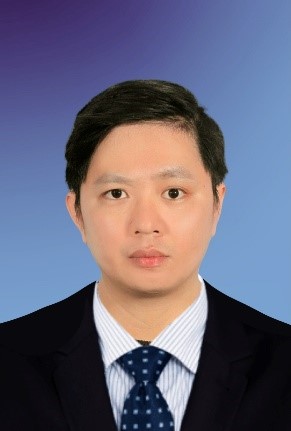}}]{Trinh Van Chien} (Member, IEEE) received the B.S. degree in electronics and telecommunications from Hanoi University of Science and Technology (HUST), Hanoi, Vietnam, in 2012, the M.S. degree in electrical and computer engineering from Sungkyunkwan University, Seoul, South Korea, in 2014, and the Ph.D. degree in communication systems from Linköping University, Linköping, Sweden, in 2020. He was a Research Associate with the University of Luxembourg, Esch-sur-Alzette, Luxembourg. He is currently with the School of Information and Communication Technology, HUST. His interest lies in convex optimization problems and machine learning applications for wireless communications and image and video processing. Dr. Chien also received the Award of Scientific Excellence in the first year of the 5G Wireless Project funded by European Union Horizon’s 2020 and the National Outstanding Young Scientist Award (Golden Globe) in Science and Technology, Vietnam, in 2023. He was an Exemplary Reviewer of IEEE WIRELESS COMMUNICATIONS LETTERS in 2016, 2017, and 2021.
\end{IEEEbiography}

\begin{IEEEbiography}[{\includegraphics[width=1in,height=1.25in,clip,keepaspectratio]{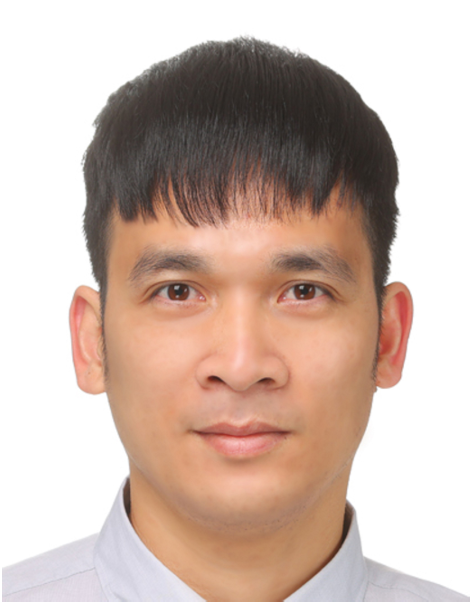}}]{Quoc-Viet Pham} (M’18, SM’23) is currently an Assistant Professor in Networks and Distributed Systems at the School of Computer Science and Statistics, Trinity College Dublin, Ireland. He earned his Ph.D. degree in Telecommunications Engineering from Inje University, Korea, in 2017. He specialises in applying convex optimisation, game theory, and machine learning to analyse and optimise cloud edge computing, wireless networks, and IoT systems. He was awarded the Korea NRF funding for outstanding young researchers for the term 2019-2024. He was a recipient of the Best Ph.D. Dissertation Award in 2017, Top Reviewer Award from IEEE Transactions on Vehicular Technology in 2020, Golden Globe Award in Science and Technology for Younger Researchers in Vietnam in 2021, IEEE ATC Best Paper Award in 2022, and IEEE MCE Best Paper Award in 2023. He was honoured with the IEEE ComSoc Best Young Researcher Award for EMEA 2023 in recognition of his research activities for the benefit of the Society. He is currently serving as an Editor for IEEE Communications Letters, IEEE Communications Standards Magazine, IEEE Communications Surveys $\&$ Tutorials, Journal of Network and Computer Applications, and REV Journal on Electronics and Communications.
\end{IEEEbiography}

\begin{IEEEbiography}[{\includegraphics[width=1in,height=1.25in,clip,keepaspectratio]{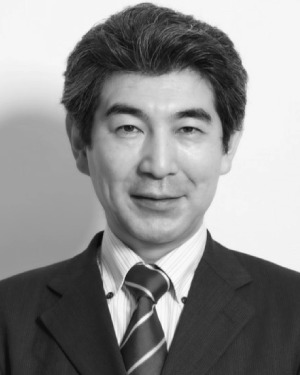}}]{Mikio Hasegawa} (Member, IEEE) received the B.Eng., M.Eng., and Dr.Eng. degrees from the Tokyo University of Science, Japan, in 1995, 1997, and 2000, respectively. From 1997 to 2000, he was a Research Fellow with the Japan Society for the Promotion of Science (JSPS). From 2000 to 2007, he was with the Communications Research Laboratory (CRL), Ministry of Posts and Telecommunications, which was reorganized as the National Institute of Information and Communications Technology (NICT), in 2004. He is currently a Professor with the Department of Electrical Engineering, Faculty of Engineering, Tokyo University of Science. His research interests include mobile networks, cognitive radio, neural networks, machine learning, and optimization techniques. He is a Senior Member of IEICE.
\end{IEEEbiography}

\begin{IEEEbiography}[{\includegraphics[width=1in,height=1.25in,clip,keepaspectratio]{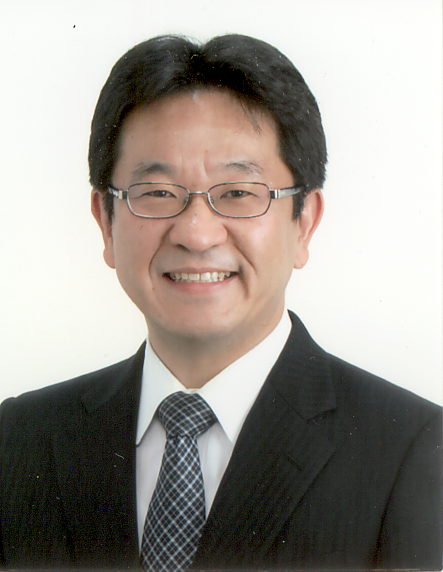}}]{Hiroo Sekiya} (Senior Member, IEEE) received B.E., M.E., and Ph.D. degrees in electrical engineering from Keio University, Yokohama, Japan, in 1996, 1998, and 2001, respectively. Since April 2001, he has been with Chiba University, Chiba, Japan, where he is currently a Professor with the Graduate School of Engineering.
His research interests include high-frequency high-efficiency tuned power amplifiers, resonant dc/dc power converters, wireless power transfer, and digital signal processing for wireless communications. He has served as a BoG member of IEEE CASS (2020-2025), AE of IEEE JESTPE,  IEEE TCAS-II, IET CDS, and so on.
\end{IEEEbiography}

\begin{IEEEbiography}[{\includegraphics[width=1in,height=1.25in,clip,keepaspectratio]{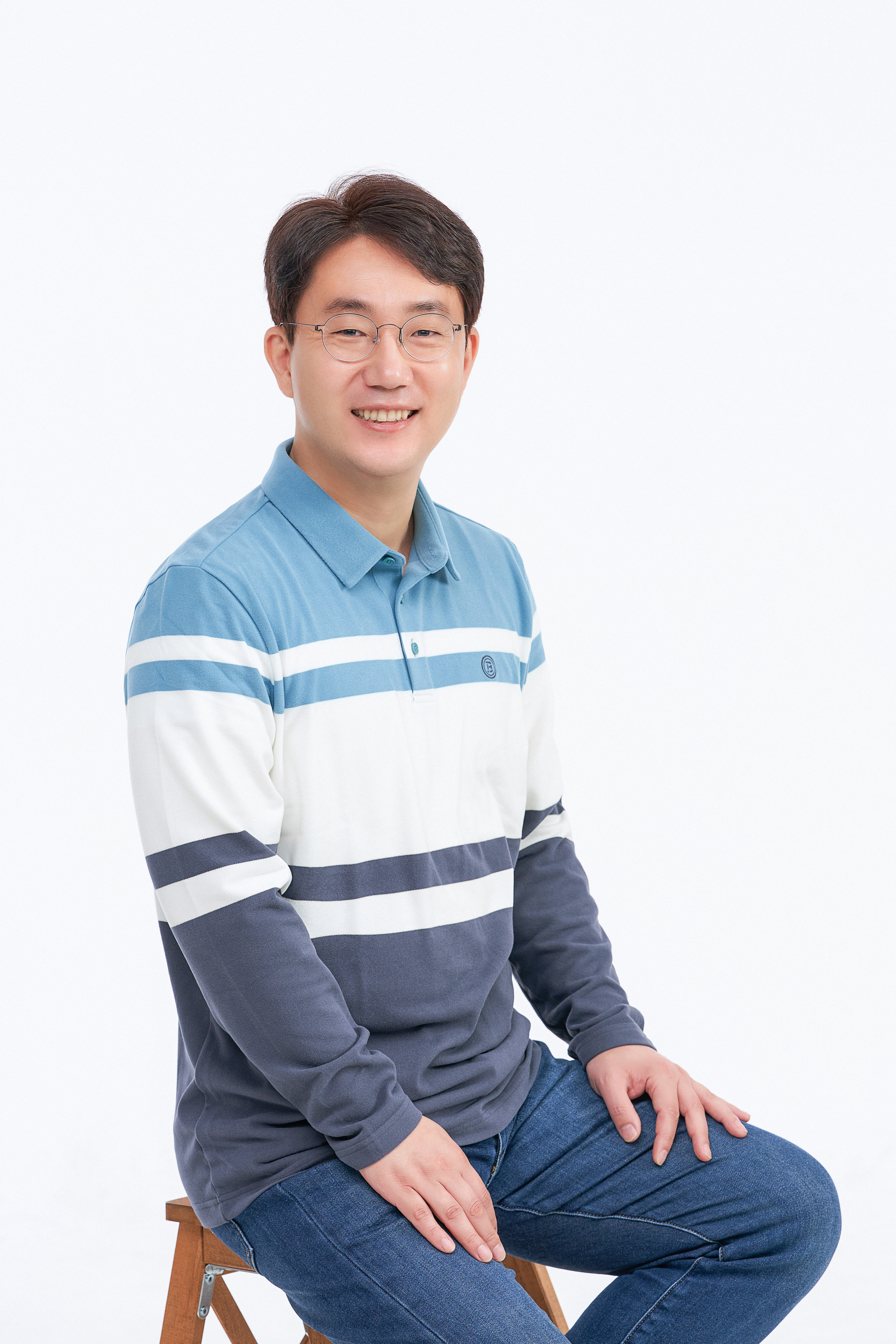}}]{Won-Joo Hwang} (Senior Member, IEEE) received the B.S. and M.S. degrees in computer engineering from Pusan National University, Busan, South Korea, in 1998 and 2000, respectively, and the Ph.D. degree in information systems engineering from Osaka University, Osaka, Japan, in 2002. From 2002 to 2019, he was employed as a full-time Professor with Inje University, Gimhae, South Korea. He is currently a full-time Professor with the School of Computer Science and Engineering, Pusan National University. His research interests include optimization theory, game theory, machine learning, and data science for wireless communications and networking.
\end{IEEEbiography}

\vfill
\end{document}